\documentclass[journal,12pt,draftcls,onecolumn]{IEEEtran}

\usepackage{amssymb}
\usepackage{amsmath}
\usepackage{epsfig}
\usepackage{epsf}
\usepackage{subfigure}
\usepackage{graphicx}
\usepackage{url}
\usepackage{algorithm}
\usepackage{algorithmic}
\usepackage{amsthm}

\def\BibTeX{{\rm B\kern-.05em{\sc i\kern-.025em b}\kern-.08em
    T\kern-.1667em\lower.7ex\hbox{E}\kern-.125emX}}

\newtheorem{corollary}{Corollary}
\newtheorem{theorem}{Theorem}
\newtheorem{definition}{Definition}

\newtheorem{lemma}{Lemma}

\newtheorem{remark}{Remark}
\author{Sina Lashgari and A. Salman Avestimehr\thanks{S. Lashgari and A. S. Avestimehr are with the School of Electrical and Computer Engineering, Cornell University, Ithaca, NY (email: sl2232@cornell.edu, avestimehr@ece.cornell.edu).
The result was presented in part at IEEE ISIT 2012 \cite{khodam}.
The research of A. S. Avestimehr and S. Lashgari was supported in part by Intel, Cisco, and Verizon (via the Video Aware Wireless Networks (VAWN) Program), and the U.S. Air Force award FA9550-11-1-0064.
}
}

\begin{document}
\title{\Large{Timely Throughput of Heterogeneous Wireless Networks: Fundamental Limits and Algorithms}}

\maketitle
\vspace{-10mm}
\begin{abstract}
The proliferation of different wireless access technologies, together with the growing number of multi-radio wireless devices suggest that the opportunistic utilization of multiple connections at the users can be an effective solution to the phenomenal growth of traffic demand in wireless networks.
In this paper we consider the downlink of a wireless network with $N$ Access Points  ({\sf AP}'s) and $M$ clients, where each client is connected to several out-of-band {\sf AP}'s, and requests delay-sensitive traffic (e.g., real-time video). We adopt the framework of Hou, Borkar, and Kumar, and study the maximum total timely throughput of the network, denoted by $C_{\text{{\sf T}}^{\text{3}}}$, which is the maximum average number of packets delivered successfully before their deadline.
Solving this problem is challenging since even the number of different ways of assigning packets to the {\sf AP}'s is $N^M$. 
We overcome the challenge by proposing a deterministic relaxation of the problem, which converts the problem to a network with deterministic delays in each link.
We show that the additive gap between the capacity of the relaxed problem, denoted by $C_{\text{det}}$, and $C_{\text{{\sf T}}^{\text{3}}}$ is bounded by $2\sqrt{N(C_{\text{det}}+\frac{N}{4})}$,
which is asymptotically negligible compared to $C_{\text{det}}$, when the network is operating at high-throughput regime. 
 In addition, our numerical results show that the actual gap between $C_{\text{{\sf T}}^{\text{3}}}$ and $C_{\text{det}}$ is in most cases much less than the worst-case gap proven analytically.
Moreover, using LP rounding methods we prove that the relaxed problem can be approximated within additive gap of $N$. 
We extend the analytical results to the case of time-varying channel states, real-time traffic, prioritized traffic, and optimal online policies. Finally, we generalize the model for deterministic relaxation to consider fading, rate adaptation, and multiple simultaneous transmissions.
\end{abstract}
\vspace{-6mm}
\begin{IEEEkeywords}
Heterogeneous wireless networks, timely throughput, scheduling, real-time traffic, network capacity.
\end{IEEEkeywords}

\section{\large Introduction} \label{introduction}

Consumer demand for data services over wireless networks has increased dramatically in recent years, fueled both by the success of online video streaming and popularity of video-friendly mobile devices like smartphones and tablets. This confluence of trends is expected to continue and lead to several fold increase in traffic over wireless networks by 2015, the majority of which is expected to be video~\cite{Cisco}. As a result, one of the most pressing challenges in wireless networks is to find effective ways to provide high volume of top quality video traffic to smartphone users.

With the evolution of wireless networks towards heterogeneous architectures, including wireless relays and femtocells, and growing number of smart devices that can connect to several wireless technologies (e.g. 3G and WiFi), it is promising that the opportunistic utilization of heterogeneous networks (where available)  can  be one of the key solutions to help cope with the phenomenal growth of video demand over wireless networks. This motivates two fundamental questions: first, how much is the ultimate capacity gain from opportunistic utilization of network heterogeneity for delay-sensitive traffic? and second, what are the optimal policies that exploit network heterogeneity for delivery of delay-sensitive traffic?

In this paper, we study these questions in the downlink of a heterogeneous wireless network with $N$ Access Points  ({\sf AP}'s) and $M$ clients. We assume that each {\sf AP} is using a distinct frequency band, 
and all {\sf AP}'s are connected to each other through a Backhaul Network (see Fig. \ref{network}), with error free links, so that we can focus on the wireless aspect of the problem. We model the wireless channels as  packet erasure channels.

We focus on real-time video streaming applications, such as video-on-demand, video conferencing, and IPTV, 
 that require tight guarantees on timely delivery of the packets.
In particular, the packets for such applications have strict-per-packet deadline; and if a packet is not delivered successfully by its deadline, it will not be useful anymore. As a result, we focus on the notion of \emph{timely throughput}, proposed in \cite{QoS}, which measures the long-term
average number of ``successful deliveries'' (i.e., the packets  delivered before the deadline) for each client as an analytical metric for evaluating both throughput and QoS for delay-constrained flows.

In this framework, time is slotted and time-slots are grouped to form intervals of length $\tau$. For each interval every client has packets to receive and the {\sf AP}'s
have to decide on a scheduling policy to deliver the packets. If a packet is  not delivered by the
end of that interval, it gets dropped by the {\sf AP}'s. Total timely throughput, $\text{{\sf T}}^{\text{3}}$, is defined as the long-term
average number of successful deliveries in the network. Our objective is then to find the maximum achievable $\text{{\sf T}}^{\text{3}}$, which we denote by $C_{\text{{\sf T}}^{\text{3}}}$ , over all possible scheduling policies.

The challenge is that for each interval, even the number of different ways of assigning packets to  {\sf AP}'s is $N^M$, which grows exponentially in the number of clients ($M$). 
For  $N=1$, \cite{QoS} provides an efficient characterization of the timely throughput region.
In fact, timely throughput region for $N=1$ can be shown to be a scaled version of a polymatroid \cite{Staticscheduling}. However, once we move beyond $N=1$, the timely throughput  region loses its polymatroidal structure which  makes the problem much more challenging.
To overcome the challenge, we propose a \emph{deterministic relaxation} of the problem, which is based on converting the problem to a network with deterministic delays for each link. As we will show in Section \ref{results}, the relaxed problem can be viewed as an assignment problem in which each {\sf AP} turns into a bin with certain capacity and each packet turns into an object which has different sizes at different bins. The relaxed problem is then to maximize the total number of objects that can be packed in the bins, denoted by $C_{\text{det}}$.

Our main contribution in this paper is two-fold. First, we prove that the gap between the solutions to the original problem ($C_{\text{{\sf T}}^{\text{3}}}$) and its relaxed version ($C_{\text{det}}$) is at most $2\sqrt {N(C_{\text{det}}+\frac{N}{4})}$. Since $N$ is typically very small (in most cases between 2-4), the above result indicates that $C_{\text{det}}$ is asymptotically equal to $C_{\text{{\sf T}}^{\text{3}}}$ as $C_{\text{{\sf T}}^{\text{3}}}\to\infty$. 
Furthermore, our numerical results demonstrate that the gap
 is in most cases much smaller than the worst-case gap that we prove analytically.
Therefore, instead of solving our main maximization problem we can solve its relaxed version, and still get a value which is very close to the optimum. 
Second, we prove that the relaxed problem can be approximated in polynomial-time (with additive gap of N) using a simple LP rounding method. This approximation is appealing as $N$ is usually limited and negligible compared to $C_{\text{det}}$.
As a result, the solution to the relaxed problem provides a scheduling policy that provably achieves a $\text{{\sf T}}^{\text{3}}$ that  is within additive gap  $N+2\sqrt {N(C_{\text{{\sf T}}^{\text{3}}}-\frac{3N}{4})}$
of $C_{\text{{\sf T}}^{\text{3}}}$ for $C_{\text{{\sf T}}^{\text{3}}}>\frac{7N}{4}$.

We also consider several extensions of the problem, including extension to time-varying channels and real-time traffic, where at the beginning of each interval clients have request for variable number of packets. We show that the aforementioned results hold in these two extensions, too. 
Moreover, we provide similar results for the case where different flows  have different priorities (different weights).
In addition, we extend the model to allow for online scheduling policies, where {\sf AP}'s are coordinated, 
 and a packet might be transmitted by arbitrary number of {\sf AP}'s.
Finally, we consider an extension to account for fading, multiple simultaneous transmissions by {\sf AP}'s and multiple simultaneous receptions by clients, and rate adaptation.

{\bf Related Work:} Although there are classical results ~\cite{Tassiulas},~\cite{Neely} on scheduling clients over time-varying channels and characterizing the average delay of service,  
in recent years there has been increasing research on serving delay-sensitive traffic over wireless networks.
This increase is due to the phenomenal increase in the volume of delay-sensitive traffic, such as video traffic.
In \cite{Bambos} packets with weights and strict deadlines have been considered; and if a packet is not delivered by its deadline, it causes a certain distortion equal to its weight. They have studied the problem of minimizing the total distortion, and have characterized the optimal control.
\cite{Puri}  considered a packet switched network where clients can get different types of service based on the amount they are willing to pay.
The problem of optimizing time averages in systems with i.i.d behavior over renewal frames  has been considered in \cite{Renewal}; and an algorithm which minimizes drift-plus-penalty ratio is developed.
Moreover, ~\cite{Srikant} has focused on minimizing the total number of expired packets, and has provided analytic results on scheduling.

However, the most related work to this paper is the work of  Hou et. al in ~\cite{QoS} in 2009, in which they have proposed a framework for jointly addressing delay, delivery ratio, and channel reliability.
For a network with one {\sf AP} and $N$ clients, the timely throughput region for the set of N clients has been fully characterized in ~\cite{QoS}; and the work has been extended to variable-bit-rate applications in ~\cite{VBR}, and time-varying channels and rate adaptation in ~\cite{Realtime}.
Although in ~\cite{QoS}-~\cite{Realtime} they provide tractable analytical results and low-complexity scheduling policies,
 the analyses are done for  only one {\sf AP}.
This paper aims to extend the results to the case of general number of {\sf AP}'s, where there is an additional challenge of how to split the packets among different {\sf AP}'s.


\section{\large Network Model and Problem Formulation}\label{model}
In this section we describe our network model and precisely describe the notion of timely throughput introduced in  \cite{QoS}. 
Finally, we formulate our problem.

\subsection{Network Model and Notion of Timely Throughput}
We consider the downlink of a network with $M$ wireless clients, denoted by $\text{ {\sf Rx}}_1,\text{ {\sf Rx}}_2,\ldots $, $\text{ {\sf Rx}}_M$, that have packet requests, and $N$ Access Points $\text {{\sf AP}}_1,\text {{\sf AP}}_2,\ldots, \text {{\sf AP}}_N$. These {\sf AP}'s have error-free links to the Backhaul Network (see Fig.1). 
In addition, time is slotted and transmissions occur during time-slots. Furthermore, the time-slots are grouped into intervals of length $\tau$, where the first interval contains the first $\tau$ time-slots, the second interval contains the second $\tau$ time-slots, and so on. Moreover, each {\sf AP} may make one packet transmission in each time-slot.

Each {\sf AP} is connected via unreliable wireless links to a subset (possibly all) of the wireless clients. 
These unreliable links are modeled as packet erasure channels that, for now, are assumed to be i.i.d over time, and have fixed success probabilities. In addition, each channel is independent of other channels in the network.
(In Section \ref{extensions} these assumptions will be relaxed to consider more general scenarios). 
The success probability of the channel between $\text {{\sf AP}}_i$ and $\text{ {\sf Rx}}_j$ is denoted by $p_{ij}$, which is the probability of successful delivery of the packet of $\text{ {\sf Rx}}_j$ when transmitted by $\text {{\sf AP}}_i$ during a time-slot. 
If there is no link between an {\sf AP}  and a client, we consider the success probability of the corresponding channel to be $0$. 
Moreover, we assume that the channels do not have interference with each other.

\begin{figure}[h!]
\centering
\subfigure[]{\label{network}\includegraphics[scale=.63,trim = 20mm 20mm 10mm 10mm]{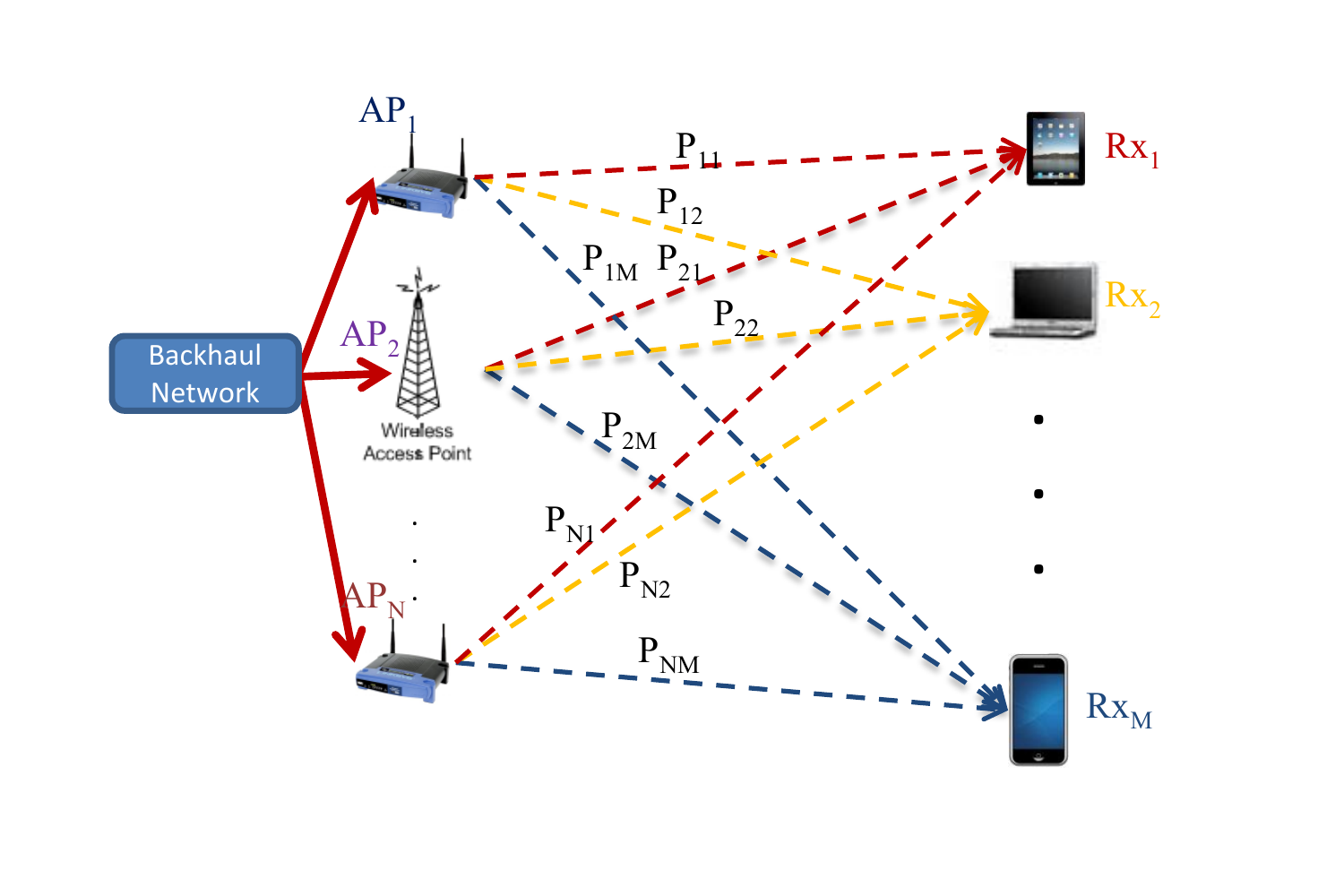}}
\subfigure[]{\label{time}\includegraphics[scale=.5,trim = 10mm 0mm 20mm 30mm]{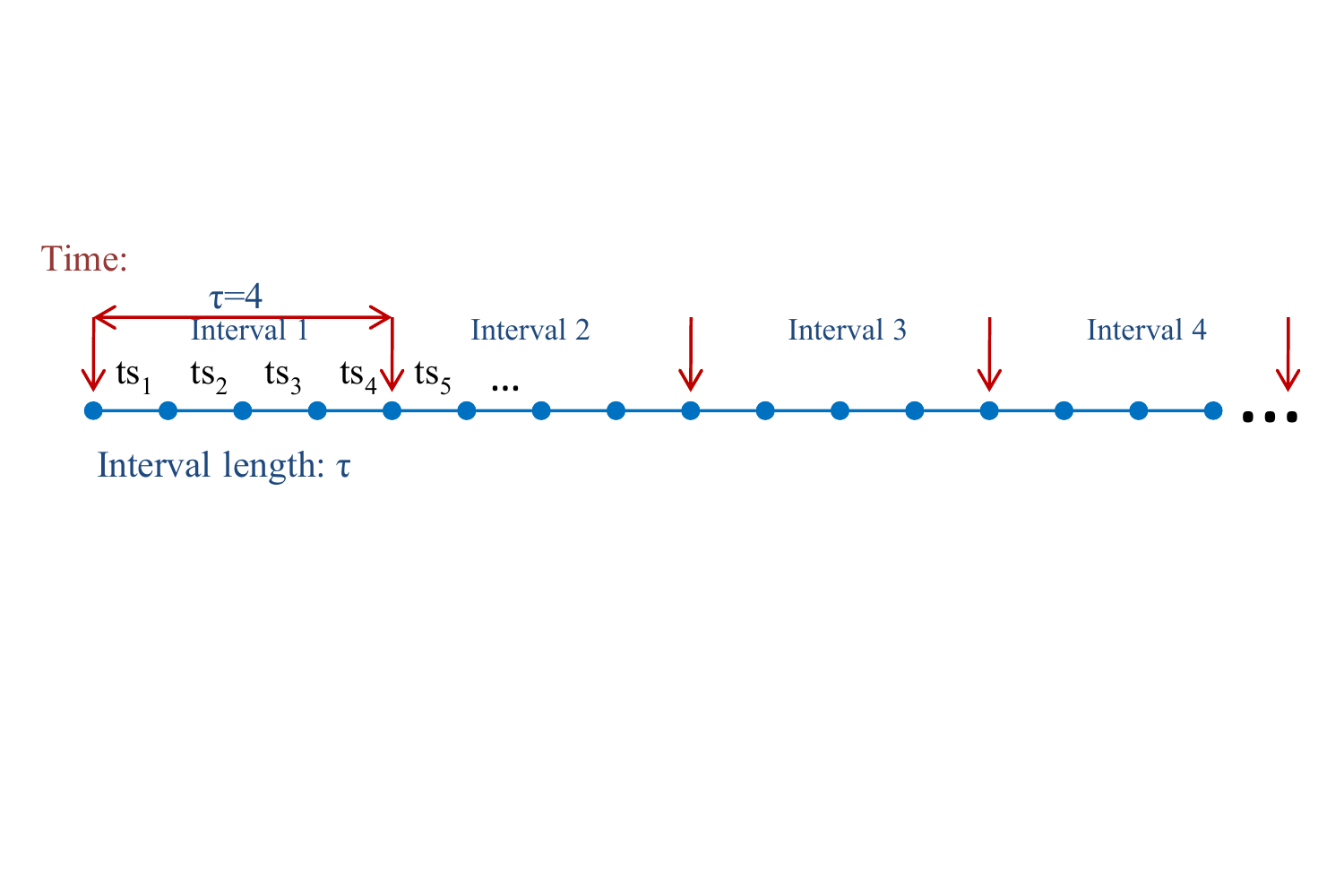}}
\caption{ Illustration of our network model. Network configuration consisting of $N$ Access points ({\sf AP}'s), $M$ wireless clients, packet erasure channels from {\sf AP}'s to the clients, and the Backhaul network is illustrated in (a).
Our time model, in which time is slotted and time-slots are grouped to form intervals of length $\tau$, is shown in (b). In this figure $\tau=4$.}
\end{figure}

For now we assume that at the beginning of each interval each client has request for a new packet. 
Right before the start of an interval, each requested packet for that interval is assigned to one of the {\sf AP}'s to be transmitted to its corresponding client.
Furthermore, during each time-slot of an interval, each {\sf AP} picks one of the packets assigned to it to transmit. 
At the end of that time-slot the {\sf AP} will know if the packet has been successfully delivered or not.
If the packet is successfully delivered, the {\sf AP} removes that packet from its buffer and does not attempt to transmit it any more. 
The packets that are not delivered by the end of the interval are dropped from the {\sf AP}'s.

\begin{definition}
The decisions on how to assign the requested packets for an interval to the {\sf AP}'s before the start of that interval, and which packet to transmit on a time-slot by each {\sf AP} are specified by a \emph{scheduling policy}. A scheduling policy $\eta$ makes the decisions causally based on the entire past history of events up to the point of decision-making.
We denote the set of all possible scheduling policies by $\mathcal S$.
\end{definition}

\begin{definition}
A \emph{static scheduling policy}, denoted by $\eta_{\text{static}}$, is a scheduling policy in which each {\sf AP} becomes responsible for serving packets of a fixed subset of clients for all intervals; and the packets of clients assigned to an {\sf AP} are served according to a fixed order. In particular, 
a static scheduling policy $\eta_{\text{static}}$ is fully specified by a pair $(\vec\Pi,\Gamma)$, in which $\vec\Pi=[\mathcal I_1,\mathcal I_2,\ldots ,\mathcal I_N]$, where $\mathcal I_i$'s partition the set $\{1,2,\ldots,M\}$, indicating how the packet of clients are assigned to {\sf AP}'s.
Furthermore, $\Gamma$ specifies the ordering for the packets assigned to each {\sf AP}.
When $\eta_{\text{static}}$ is implemented, each {\sf AP} is responsible for serving packet of the clients assigned to it by $\vec\Pi$; and each {\sf AP} persistently transmits a packet until it is delivered successfully, before moving on to the packet of the client with the immediate lower rank in the ordering specified by $\Gamma$.
\end{definition}

\begin{definition}
A static scheduling policy is called \emph{greedy}, and denoted by $\eta_{\text{g-static}}$, if the order of clients specified by $\Gamma$ is according to the success probabilities of channels from {\sf AP} to those clients, in decreasing order.    
  \end{definition}

Assume that a particular scheduling policy $\eta$ is chosen. For any interval $r$ ($r\in \mathbb N$), let $\vec N(r,\eta)\triangleq [N_1(r,\eta), N_2(r,\eta),\ldots, N_M(r,\eta)]$ denote the vector of $M$ binary elements whose $j^{th}$ element $N_j(r,\eta)$ is $1$ if client $\text{ {\sf Rx}}_j$ has successfully received a packet during the $r^{th}$ interval, and $0$ otherwise. 
When using scheduling policy $\eta$, the total timely throughput, denoted by $\text{{\sf T}}^{\text{3}}(\eta)$, is defined as 
\begin{equation}
\text{{\sf T}}^{\text{3}}(\eta)\triangleq \limsup_{r\to\infty} \frac{\sum_{k=1}^{r}\sum_{j=1}^{M}  N_j(k,\eta)}{r}.\label{liminf}
\end{equation}
In simpler words, $\text{{\sf T}}^{\text{3}}(\eta)$ is the long-term average number of successful deliveries in the entire network. 
Similarly,  the timely throughput of $\text{ {\sf Rx}}_j$, denoted by $R_j(\eta)$, is defined as 
\begin{equation}
R_j(\eta)\triangleq \limsup_{r\to\infty} \frac{\sum_{k=1}^{r} N_j(k,\eta)}{r},\quad j=1,2,\ldots ,M.\label{liminfrate}
\end{equation}
Therefore, $R_j(\eta)$ is the long-term average number of successful deliveries for the $j^{th}$ client. 
Further, we denote the vector of all $R_j(\eta)$'s by $\vec R(\eta)$, where we have $\vec R(\eta)\triangleq [R_1(\eta),R_2(\eta),\ldots ,\\R_M(\eta)]$.
Therefore, the capacity region for timely throughput of $M$ clients in the network is defined as 
$\mathcal C\triangleq \{\vec R(\eta): \eta\in\mathcal S\}.$

\subsection{Main Problem}
Our objective is to find the maximum achievable total timely throughput, denoted by $C_{\text{{\sf T}}^{\text{3}}}$. 
More precisely, our optimization problem is 
\begin{equation}\label{main}
\text{\emph{Main Problem (MP)}:}\qquad C_{\text{{\sf T}}^{\text{3}}}\triangleq\sup_{\eta\in\mathcal S}\text{{\sf T}}^{\text{3}}(\eta).
\end{equation}
Later in Section \ref{weightedt3} we will consider the problem of finding the maximum weighted total timely throughput $\sum_{j=1}^{M}\omega_j R_j(\eta)$ and its corresponding policy $\eta$; but for now we focus on the problem in the case that $\omega_1=\omega_2=\ldots =\omega_M=1$.

\subsection{Remarks on the Main Problem}

As we state later  in Lemma \ref{lemorder} in Section \ref{theorem1}, $C_{\text{{\sf T}}^{\text{3}}}$ can be achieved using a greedy static scheduling policy.
Therefore,  the optimization in (\ref{main}) can be limited to  finding the partition $\vec\Pi$ such that the corresponding $\eta_{\text{g-static}}$ maximizes $\text{{\sf T}}^{\text{3}}(\eta_{\text{g-static}})$.
However, this is still quite challenging.
In fact, the number of possible greedy static scheduling policies to consider is $N^M$, which grows exponentially in $M$.

In \cite{QoS} Hou et al. have found the timely throughput region for $N=1$, and have shown that it is a scaled version of a polymatroid \cite{Staticscheduling}.
However, when going from one {\sf AP} to several {\sf AP}'s the problem changes quintessentially: the timely throughput region loses its polymatroidal structure, which makes the problem much more challenging\footnote[1]{Example: Let $N=M=2, \tau=1$, and $p_{11}=p_{12}=p_{21}=p_{22}=1/2$.
In this case, the region is the convex hull of three points $(3/4,0),(1/2,1/2),(0,3/4)$. Therefore, no scaled version of the capacity region along its axes can be a polymatroid.}. In this case the timely throughput  region is a general polytope with (possibly) exponential number of corner points (corresponding to exponential number of ways of partitioning the clients between the {\sf AP}'s). 

\section{\large Deterministic Relaxation and Statement of Main Results}\label{results}
In this section we first explain the intuition behind proposing our relaxation scheme and formulate the relaxed problem.
Then, we state the main results.
\subsection{Deterministic Relaxation} 
In the system model  we assumed channel success probability $p_{ij}$ between $\text {{\sf AP}}_i$ and  $\text{ {\sf Rx}}_j$, $i=1,2,\ldots, N$, $j=1,2,\ldots, M$.
For now, suppose that $\tau=\infty$, $\text {{\sf AP}}_i$ has only one packet, and wants to transmit that packet to client $j$. 
Thus, $\text {{\sf AP}}_i$ persistently sends that packet to client $j$ until the packet goes through. The number of time-slots expended for this packet to be delivered is a Geometric random variable $G_{ij}$ where $\Pr(G_{ij}=k)=p_{ij}(1-p_{ij})^{k-1},\quad k\in \mathbb N$.
We know that $E[G_{ij}]=\frac{1}{p_{ij}}$, and without any deadline, it takes $\frac{1}{p_{ij}}$ time-slots on average for packet of $\text{ {\sf Rx}}_j$ to be delivered when transmitted by $\text {{\sf AP}}_i$.

Therefore, a memory-less erasure channel with success probability $p_{ij}$ can be viewed as a pipe with variable delay which takes a packet from  $\text {{\sf AP}}_i$ and gives it to $\text{ {\sf Rx}}_j$ according to that variable delay. The probability distribution of the delay is Geometric with parameter $p_{ij}$.

To simplify the problem, we proposed to relax each channel into a bit pipe with deterministic delay equal to the inverse of its success probability. 
Therefore, for any packet of $\text{ {\sf Rx}}_j$, when assigned to  $\text {{\sf AP}}_i$ for transmission, we associate a fixed size of $\frac{1}{p_{ij}}$ to that packet. 
This means that each packet assigned to an {\sf AP} can be viewed as an object with a size, where the size varies from one {\sf AP} to another; because $\frac{1}{p_{ij}}$'s for different $i$'s are not necessarily the same. 
On the other hand, we know that each {\sf AP} has $\tau$ time-slots during each interval to send the packets that are assigned to it.
Therefore, we can view each {\sf AP} during each interval as a bin of capacity $\tau$. 
Therefore, our new problem is a packing problem; i.e., we want to see over all different assignments of objects to bins what the maximum number of objects  is that we can fit in those $N$ bins of capacity $\tau$.
We denote this maximum possible number of packed objects by $C_{\text{det}}$.
More precisely, if   we  define  $x_{ij}$  as  the  $0 - 1$   variable  which equals  $1$  if packet of client  $j$  is  assigned  to   $\text {{\sf AP}}_i$,  and  $0$  otherwise,  then  the relaxed problem  can be  formulated as following.
\begin{align}
\text{\emph{Relaxed Problem (RP):}}\qquad C_{\text{det}}\triangleq\max\quad &\sum_{i=1}^{N}\sum_{j=1}^{M}x_{ij}\label{relaxed}\\
s.t.\quad &\sum_{j=1}^{M}\frac{x_{ij}}{p_{ij}}\leq \tau\quad i=1,2,\ldots, N \label{magool1} \\
&\sum_{i=1}^{N}x_{ij}\leq 1\quad j=1,2,\ldots, M \label{magool2} \\
& x_{ij}\in\{0,1\}. \label{magool3}
\end{align}

\subsection{Main Results}
We now present the main results of the paper via two Theorems.
Theorem \ref{maintheorem} bounds  the gap between the solution to the main problem (\ref{main}) and its relaxation (\ref{relaxed}).
Furthermore, Theorem \ref{algorithm} provides a performance guarantee to the approximation algorithm for the relaxed problem.
The proofs of the two Theorems are provided in Section \ref{theorem1} and Section \ref{theorem2}.
\begin{theorem}\label{maintheorem}
Let $C_{\text{{\sf T}}^{\text{3}}}$ denote the value of the solution to our main problem in (\ref{main}). Also, let $C_{\text{det}}$ denote the value of the solution to our relaxed problem in (\ref{relaxed}). We have
\begin{equation}\label{result}
C_{\text{det}}-2\sqrt {N(C_{\text{det}}+\frac{N}{4})}<C_{\text{{\sf T}}^{\text{3}}}<C_{\text{det}}+N.
\end{equation}
\end{theorem}

\begin{remark}
The right part of the inequality in (\ref{result}) suggests that $C_{\text{{\sf T}}^{\text{3}}}-C_{\text{det}}$ can be no more than $N$.
But the number of {\sf AP}'s $N$ is limited and is usually around $2,3,$ or $4$. 
Therefore, as $C_{\text{det}}\to\infty$ $\frac{N}{C_{\text{det}}}\to 0$.
Moreover, the left  inequality in Theorem \ref{maintheorem} suggests that $C_{\text{det}}-C_{\text{{\sf T}}^{\text{3}}}$  becomes negligible compared to $C_{\text{det}}$ as $C_{\text{det}}\to\infty$. 
In addition, the inequalities in Theorem \ref{maintheorem} imply that as $C_{\text{{\sf T}}^{\text{3}}}\to\infty$, $C_{\text{det}}\to\infty$, too.
Therefore, $\frac{C_{\text{det}}}{C_{\text{{\sf T}}^{\text{3}}}}\to 1$, as $C_{\text{{\sf T}}^{\text{3}}}\to \infty$. 
Hence, the bounds in Theorem \ref{maintheorem} suggest the asymptotic optimality of solving $C_{\text{det}}$ instead of $C_{\text{{\sf T}}^{\text{3}}}$.
\end{remark}

Theorem \ref{maintheorem}, basically bounds the gap between $C_{\text{\sf T}^{\text{3}}}$ and $C_{\text{det}}$. However, a remaining question is: if we run the system based on the greedy static scheduling policy which uses the assignment proposed by the solution to the relaxed problem, how much do we lose in terms of total timely throughput compared to $C_{\text{{\sf T}}^{\text{3}}}$? The following corollary which is proved in Appendix \ref{corollary} addresses this question.

\begin{corollary}\label{run}
Assume $C_{\text{{\sf T}}^{\text{3}}}\geq \frac{7N}{4}$. Let $\vec\Pi_{\text{det}}$ denote the assignment of clients to AP's suggested by the solution to the relaxed problem (\ref{relaxed}), and $\eta_{\text{det}}$ be the corresponding greedy static scheduling policy. Then, we have
\begin{equation}
C_{\text{{\sf T}}^{\text{3}}}-N-2\sqrt {N(C_{\text{{\sf T}}^{\text{3}}}-\frac{3N}{4})}\leq ||\vec R( \eta_{\text{det}})||_1\leq C_{\text{{\sf T}}^{\text{3}}}\label{yilu}.\nonumber
\end{equation}
\end{corollary}

\begin{remark}\label{tightness}
As we prove in Appendix \ref{tightnessproof} the upper bound given in the right inequality in Theorem \ref{maintheorem} is tight. Furthermore, the lower bound given in the left inequality of Theorem 1 is tight in terms of order, i.e., there exists a network configuration and a positive constant $k$ for which $C_{\text{det}}-C_{\text{{\sf T}}^{\text{3}}}>k\sqrt{NC_{\text{det}}}.$
\end{remark}

\begin{figure}
\centering
\subfigure[]{\label{conf}\includegraphics[scale=.4,trim = 50mm 30mm 40mm 100mm]{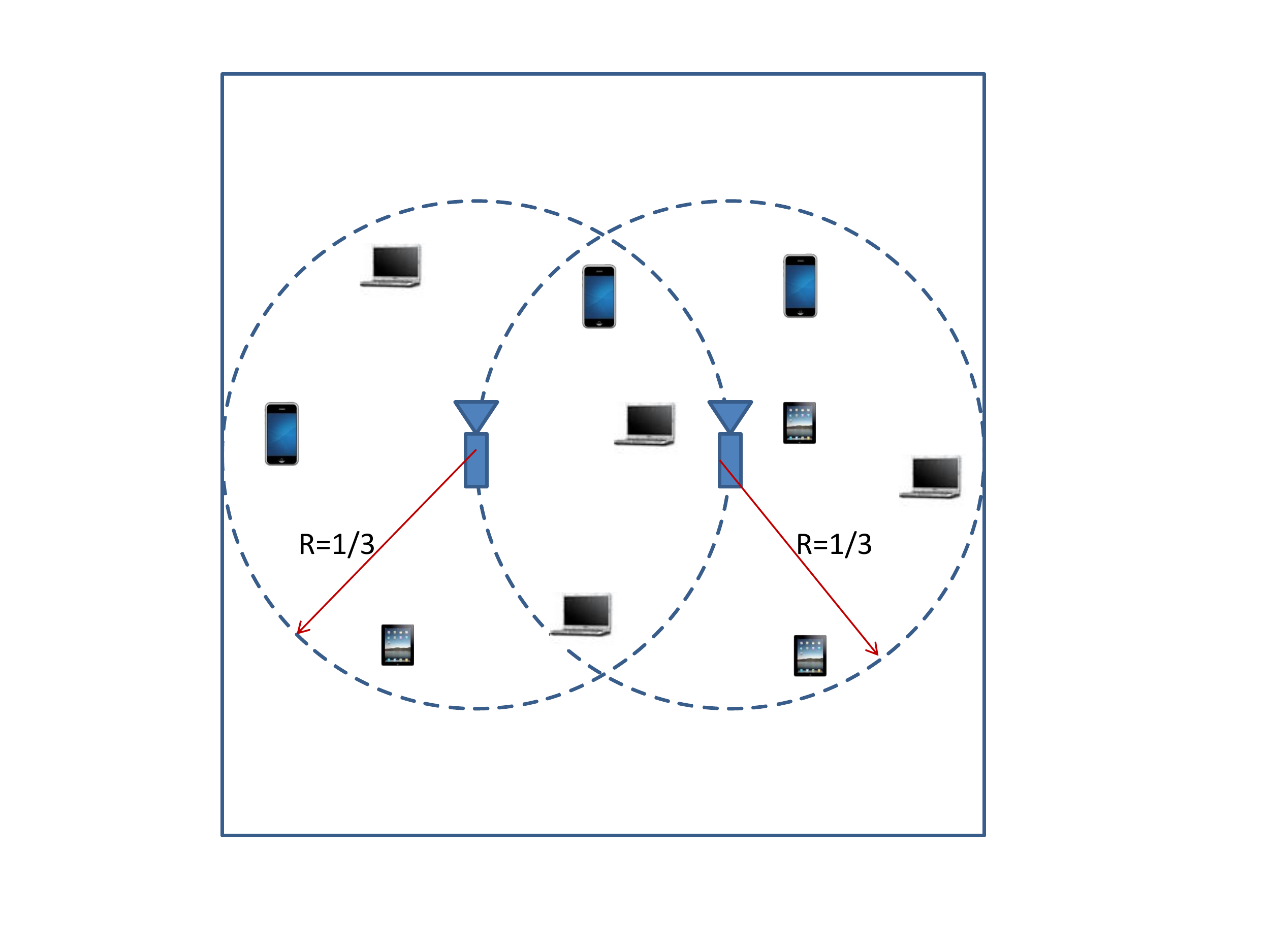}}
\subfigure[]{\label{numres}\includegraphics[scale=.6,trim = 40mm 90mm 50mm 100mm]{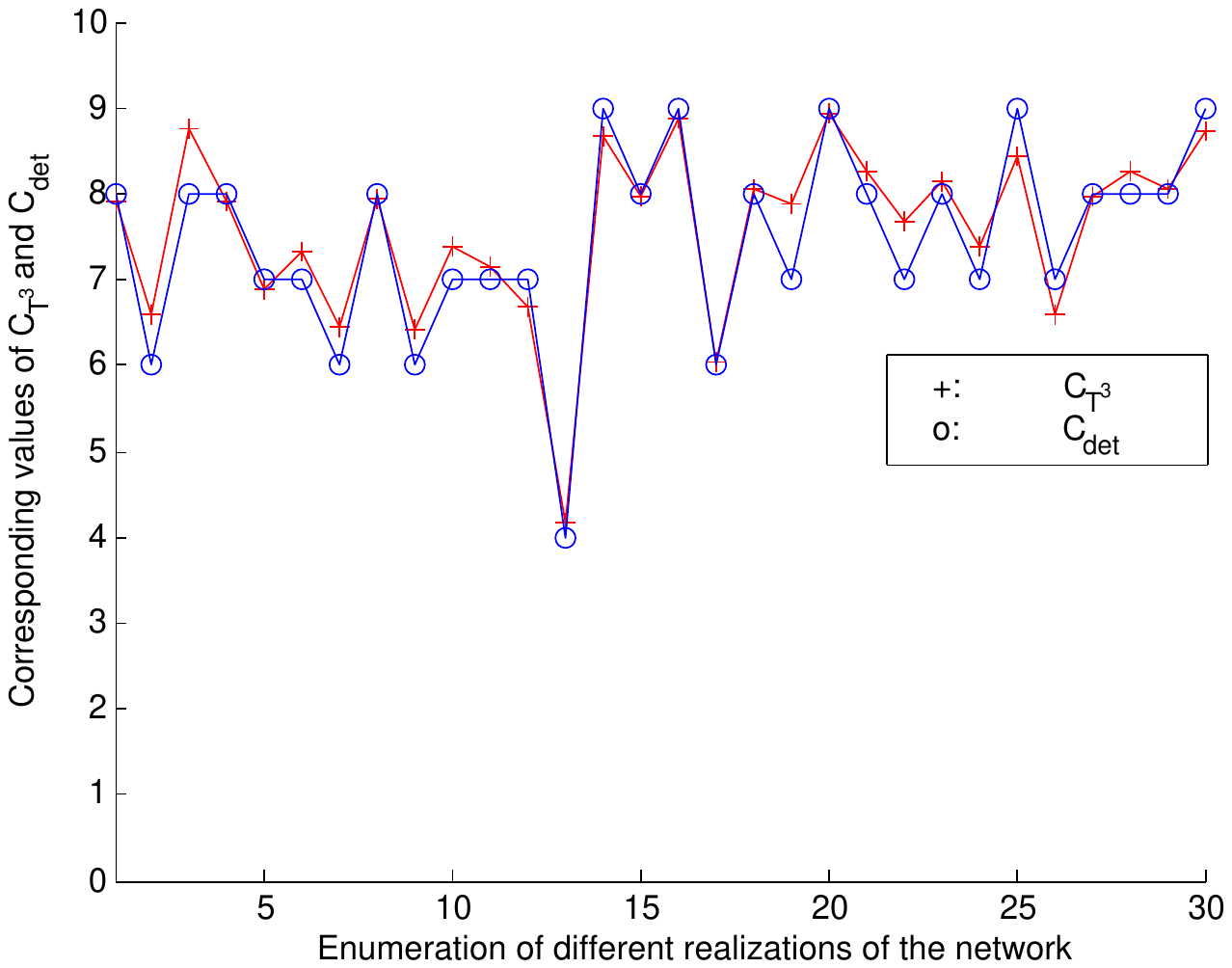}}
\caption{Numerical analysis for the gap between $C_{\text{{\sf T}}^{\text{3}}}$ and $C_{\text{det}}$ for the case of two {\sf AP}'s with coverage radius $\frac{1}{3}$, $10$ randomly located wireless clients, and intervals of length $\tau=15$. 
 (a) illustrates the network configuration, where erasure probability of a channel is proportional to the distance between the {\sf AP} and the corresponding receiver ($\text{erasure}=\min\{\frac{\text{distance}}{1/3}, 1\}$).
(b) demonstrates the numerical results for the gap for $30$ different realizations of the network, where each realization is constructed from a random and uniform location of clients in the network. Each `+' indicates the value of $C_{\text{{\sf T}}^{\text{3}}}$ for each realization, while `o' indicates the value of $C_{\text{det}}$ for the same realization.}
\end{figure}

\begin{remark}\label{smallgap}
The bounds in Theorem 1 are worst-case bounds, and via numerical analysis we observe that the gap between the original problem and its relaxation is in most cases much smaller.
Therefore, the solution to the relaxed problem tracks the solution to the main problem very well, even for a limited number of clients.  
To illustrate this, consider the network configuration in Figure~\ref{conf}, where there are two {\sf AP}'s with coverage radius $\frac{1}{3}$, and $10$ clients which are uniformly and randomly located in the coverage area of the two {\sf AP}'s.
The erasure probability of the channel between a client and an {\sf AP} is proportional to the distance ($\text{erasure}=\min\{\frac{\text{distance}}{1/3}, 1\}$); and $\tau=15$.
For 30 different realizations of this network, $C_{\text{{\sf T}}^{\text{3}}}$ and $C_{\text{det}}$ have been calculated, and plotted in Figure~\ref{numres} (detailed numerical results are provided in Section \ref{numerical}).
The numerical results suggest that even for small-scale networks $C_{\text{det}}$ is usually very close to $C_{\text{{\sf T}}^{\text{3}}}$.
\end{remark}

So far, we have shown by Theorem \ref{maintheorem} that by considering the relaxed problem (RP) we do not lose much in terms of total timely throughput capacity.  
Nevertheless, in order for the relaxation to be useful there should be a way to solve the relaxed problem efficiently.
The following algorithm approximates the solution to the relaxed problem (RP).

\begin{algorithm}
\caption{}
\begin{algorithmic}\label{apxalgorithm} 

\STATE Input: $N,M,\tau,$ and  $p_{ij}$ for $i=1,2,\ldots, N$ and $j=1,2,\ldots, M$.
\STATE Find $\bold{x^{*}}=[x^{*}_{ij}]_{N\times M}$, a basic optimal solution to the LP-relaxation of RP in (\ref{relaxed}).
\STATE  Output $\lfloor x^{*}_{ij}\rfloor$ (rounded down version of the elements of $\bold{x^{*}}$) for $i=1,2,\ldots, N$ and $j=1,2,\ldots, M$.
\end{algorithmic}
\end{algorithm}

The next Theorem, which is proved in  Section \ref{theorem2}, demonstrates that 
Algorithm \ref{apxalgorithm} approximates the  relaxed problem efficiently.

\begin{theorem}\label{algorithm}
Suppose that $\bold{x^{*}}$ is a basic optimal solution to the LP relaxation of RP. We have
\begin{equation}
C_{\text{det}}-\sum_{i=1}^{N}\sum_{j=1}^{M}\lfloor \bold {x^{*}_{ij}}\rfloor\leq N.\nonumber
\end{equation}
\end{theorem}

\begin{remark}
Finding a basic optimal solution to a linear program efficiently is straightforward, and is discussed in \cite{Basicopt}.
According to Theorem \ref{algorithm} if we find a basic optimal solution to LP relaxation of (\ref{relaxed}), and round down that solution to get integral values, the result will deviate from the optimal solution by at most $N$.
Since $N$ is typically very small (in most cases between 2-4), algorithm \ref{algorithm} performs well in approximating the solution to the Relaxed Problem (RP).
\end{remark}

\begin{remark}
The relaxed problem in (\ref{relaxed}) is a special case of the well-known Maximum Generalized
Assignment Problem (GAP).
There is a large body of literature on GAP; 
and its special cases capture many combinatorial optimization problems, having several
applications in computer science and operations research.
Even the special case of GAP in (\ref{relaxed}) is APX-hard \cite{MKP}, meaning that there is no polynomial-time approximation scheme (PTAS) for it. However, there are
several approximation algorithms for GAP, including \cite{MKP}, \cite{Goemans}. In particular,
 \cite{MKP}, based on a modification of the work in
\cite{Shmoys}, has proposed a 2-approximation algorithm for GAP; and \cite{Goemans} has
proposed an LP-based $\frac{e}{e-1}$-approximation algorithm.
The performance guarantees in the literature are concerned with multiplicative gap. 
However, our result in Theorem \ref{algorithm} suggests an additive gap performance guarantee of $N$ for the special case of GAP presented in (\ref{relaxed}).
 Since $N$ (the number of access points) is typically very small, this provides a tighter approximation guarantee for our problem of interest.
\end{remark}

\section{\large Analysis of Approximation Gap (Proof of Theorem \ref{maintheorem})}\label{theorem1}
In order to prove Theorem \ref{maintheorem}, we first state Lemma \ref{lemorder} which is proved in Appendix \ref{LemmaInv}.

\begin{lemma}\label{lemorder}
$C_{\text{{\sf T}}^{\text{3}}}$ can be achieved using a greedy static scheduling policy.
\end{lemma}

Lemma 1 shows there is a scheduling policy which uses the same assignment and ordering of the packets for all intervals, and achieves $C_{\text{{\sf T}}^{\text{3}}}$.
The result in Lemma 1 is intuitive, and is a consequence of time-homogeneity of the system (Lemma 1 is also true for the time-varying channel model where channels are modeled by FSMC).
In fact, Lemma \ref{lemorder} allows us to focus on only one interval, and then to maximize the expected number of deliveries over that interval.

However, the main challenge lies in how to optimally assign the packets to {\sf AP}'s in order to maximize the expected number of deliveries.
But once the assignment is specified, the optimal ordering is trivial according to Lemma \ref{lemorder}.
We now use Lemma \ref{lemorder} in order to prove the right side of the inequality in Theorem \ref{maintheorem}.

\subsection{Proof of $C_{\text{{\sf T}}^{\text{3}}}<C_{\text{det}}+N$}
By Lemma \ref{lemorder} it is sufficient to prove that for any greedy static scheduling policy $\eta_{\text{g-static}}$,
$\text{\sf T}^{\text{3}}(\eta_{\text{g-static}})<C_{\text{det}}+N.$
Suppose an arbitrary greedy static scheduling policy $\eta_{\text{g-static}}$ with the corresponding partition $\vec\Pi_{\text{g-static}}=[\mathcal{I}_1,\mathcal{I}_2,\ldots ,\mathcal{I}_N]$ and ordering $\Gamma_{\text{g-static}}$ is implemented. 
By (\ref{liminf}) we know that 
\begin{equation}\label{T3RJ}
\text{\sf T}^{\text{3}}(\eta_{\text{g-static}})=\limsup_{r\to\infty}\frac{\sum_{k=1}^{r}\sum_{j=1}^{M}N_j(k,\eta_{\text{g-static}})}{r}.
\end{equation}
On the other hand, by (\ref{liminfrate}) we know that for $j\in[1:M]$,
$R_j(\eta_{\text{g-static}})=\limsup_{r\to\infty}\frac{\sum_{k=1}^{r}N_j(k,\eta_{\text{g-static}})}{r}.$
Let $Y_i$ denote the random variable for the number of successful deliveries by  $\text {{\sf AP}}_i$ during one interval, when $\eta_{\text{g-static}}$ is implemented;  in other words,
$Y_i\triangleq \sum_{j\in \mathcal I_i}^{}N_j(1,\eta_{\text{g-static}}),\quad i\in[1:N].$
Since a greedy static scheduling policy is implemented and channels are i.i.d over time, by LLN,
\begin{align}
\sum_{i=1}^{N}E[Y_i]&=\sum_{i=1}^{N}\limsup_{r\to\infty}\frac{\sum_{j\in \mathcal I_i}^{}\sum_{k=1}^{r}N_j(k,\eta_{\text{g-static}})}{r}
=\sum_{i=1}^{N}\sum_{j\in\mathcal I_i}^{}R_j(\eta_{\text{g-static}})\nonumber\\
&=\sum_{j=1}^{M}R_j(\eta_{\text{g-static}})=\sum_{j=1}^{M}\lim_{r\to\infty}\frac{\sum_{k=1}^{r}N_j(k,\eta_{\text{g-static}})}{r}\nonumber\\
&=\lim_{r\to\infty}\frac{\sum_{k=1}^{r}\sum_{j=1}^{M}N_j(k,\eta_{\text{g-static}})}{r}=\text{\sf T}^{\text{3}}(\eta_{\text{g-static}}).\label{LLN2}
\end{align}
Define $q_i\triangleq|\mathcal I_i|$, and denote the enumeration of clients assigned to $\text{{\sf AP}}_i$ by
$\{\mathcal I_i(1),\mathcal I_i(2),\ldots,\mathcal I_i(q_i)\}$,
where the enumeration is according to the channel success probabilities of different clients in $\mathcal I_i$.
Let $G_{ij}$ be a geometric random variable with parameter $p_{ij}, i\in[1:N], j\in [1:M]$.
Then, it is easy to see that 
\begin{equation}
Y_i= \max\quad k\quad s.t. \quad \sum_{j=1}^{k} G_{i\mathcal I_i(j)}\leq \tau, \quad i\in\{1,2,\ldots ,N\},k\leq q_i,\nonumber
\end{equation}
since $\eta_{\text{g-static}}$ persistently sends a packet until it is delivered, or the interval is over.
Define 
\begin{equation}  
l_i\triangleq \max \quad \hat l\quad s.t.\quad \sum_{j=1}^{\hat l}1/p_{i\mathcal I_i(j)}\leq \tau,\quad \hat l\leq q_i.\nonumber
\end{equation}
Therefore, $l_i$ is the maximum number of objects that fit into bin of capacity $\tau$ when the channels are relaxed and clients in $\mathcal I_i$ are assigned to ${\sf AP}_i$.
The following lemma (for which the proof is provided in Appendix \ref{rightlemma}) relates $l_i$ to $Y_i$.
\begin{lemma}\label{lemmaup}
Let $\tau \in \mathbb N$ and $G_1, G_2, \ldots, G_q$ be independent geometric random variables with parameters $p_1, p_2, \ldots, p_q$ respectively, such that $1\geq p_1 \geq p_2 \geq \ldots \geq p_q \geq 0$. Also define
$l\triangleq \max \hat l\quad s.t.\quad \sum_{i=1}^{\hat l}1/p_i\leq \tau$,
 and
$Y\triangleq \max i\quad s.t. \quad \sum_{j=1}^{i} G_j\leq \tau, \quad i\in\{1,2,\ldots ,q\}.$
Then,
$E[Y]<l+1.$
\end{lemma}

Hence, 
\begin{align}
 \text{\sf T}^{\text{3}}(\eta_{\text{g-static}}) &\stackrel{(a)}{=} \sum_{i=1}^{N}E[Y_i] \stackrel{(b)}{<} \sum_{i=1}^{N}(l_i+1)
\stackrel{(c)}{\leq}  C_{\text{det}}+N.\nonumber
\end{align}
where (a) follows from (\ref{LLN2});
(b) follows from Lemma \ref{lemmaup};
and (c) follows from the fact that $\sum_{i=1}^{N}l_i$ is the value of the objective function in (\ref{relaxed}) for a feasible solution.
Hence the proof of the right inequality in Theorem \ref{maintheorem} is complete.

\subsection{Proof of $C_{\text{det}}-2\sqrt {N(C_{\text{det}}+\frac{N}{4})}<C_{\text{{\sf T}}^{\text{3}}}$}

Consider the assignment proposed by the solution to the relaxed problem in (\ref{relaxed}), where the clients that are not assigned to any {\sf AP} for transmission are now assigned to {\sf AP}'s arbitrarily. 
Let $\vec\Pi_{\text{g-static}}^{\text{det}}=[\mathcal{I}_1^{\text{det}},\mathcal{I}_2^{\text{det}},\ldots ,\mathcal{I}_N^{\text{det}}]$ denote the resulting partition, and also let $\eta_{\text{g-static}}^{\text{det}}$ denote the corresponding greedy static scheduling policy.
Therefore, we have
$\text{\sf T}^{\text{3}}(\eta_{\text{g-static}}^{\text{det}})\leq C_{\text{\sf T}^{\text{3}}}\nonumber.$
So, it is sufficient to prove that $C_{\text{det}}-2\sqrt {N(C_{\text{det}}+\frac{N}{4})}<\text{{\sf T}}^{\text{3}}(\eta_{\text{g-static}}^{\text{det}})$.
Let $Y_i^{\text{det}}$ denote the random variable indicating the number of successful deliveries by  $\text {{\sf AP}}_i$ during one interval, when $\eta_{\text{g-static}}^{\text{det}}$ is implemented, $i=1,2,\ldots ,N$.
 With the same argument as in part A we have
$\text{{\sf T}}^{\text{3}}(\eta_{\text{g-static}}^{\text{det}})=\sum_{i=1}^{N}E[Y_i^{\text{det}}].$
Therefore, it is sufficient to prove that 
$C_{\text{det}}-2\sqrt {N(C_{\text{det}}+\frac{N}{4})}<\sum_{i=1}^{N}E[Y_i^{\text{det}}].$
Define $q_i=|\mathcal I_i^{\text{det}}|$; and denote the enumeration of clients assigned to $\text{{\sf AP}}_i$ by
$\{\mathcal I_i^{\text{det}}(1),\mathcal I_i^{\text{det}}(2),\ldots,\mathcal I_i^{\text{det}}(q_i)\}$,
where the enumeration is according to the channel success probabilities of different clients in $\mathcal I_i^{\text{det}}$.
Further, let $G_{ij}$ be a geometric random variable with parameter $p_{ij},\quad i=1,2,\ldots, N,\quad j=1,2,\ldots,M$.
Then, it is easy to see that 
\begin{equation}
Y_i^{\text{det}}= \max\quad k\quad s.t. \quad \sum_{j=1}^{k} G_{i\mathcal I_i^{\text{det}}(j)}\leq \tau, \quad i\in\{1,2,\ldots ,N\},k\leq q_i,\nonumber
\end{equation}
since $\eta_{\text{g-static}}^{\text{det}}$ persistently sends a packet until it is delivered, or the interval is over.
Also define 
\begin{equation}  
l_i^{\text{det}}\triangleq \max \quad \hat l\quad s.t.\quad \sum_{j=1}^{\hat l}1/p_{i\mathcal I_i^{\text{det}}(j)}\leq \tau,\quad \hat l\leq q_i.\nonumber
\end{equation}
Therefore, $l_i^{\text{det}}$ is the maximum number of objects that fit into a bin of capacity $\tau$ when the channels are relaxed and clients in $\mathcal I_i^{\text{det}}$ are assigned to ${\sf AP}_i$. The following lemma (which is proved in Appendix \ref{leftlemma}) relates $l_i^{\text{det}}$ to $Y_i^{\text{det}}$.

\begin{lemma}\label{lemmadown}
Let $\tau \in \mathbb N$ and $G_1, G_2, \ldots, G_q$ be independent geometric random variables with parameters $p_1, p_2, \ldots, p_q$ respectively, such that $1\geq p_1 \geq p_2 \geq \ldots \geq p_q \geq 0$. Also define
$l\triangleq \max \hat l\quad s.t.\quad \sum_{i=1}^{\hat l}1/p_i\leq \tau$
 and
$Y\triangleq \max i\quad s.t. \quad \sum_{j=1}^{i} G_j\leq \tau, \quad i\in\{1,2,\ldots ,q\}.$
Then, 
$l-2\sqrt{l+\frac{1}{4}}<E[Y].$
\end{lemma}
Hence,
\begin{align}
C_{\text{{\sf T}}^{\text{3}}} & \geq \sum_{i=1}^{N}E[Y_i^{\text{det}}]  \stackrel{(a)}{>}  \sum_{i=1}^{N} l_i^{\text{det}}-2\sum_{i=1}^{N}\sqrt{l_i^{\text{det}}+\frac{1}{4}}\nonumber \\
& \stackrel{(b)}{\geq} \sum_{i=1}^{N} l_i^{\text{det}}-2\sqrt {N(\sum_{i=1}^{N}l_i^{\text{det}}+\frac{N}{4})}  = C_{\text{det}}-2\sqrt {N(C_{\text{det}}+\frac{N}{4})},\nonumber
\end{align}
where (a) follows from Lemma \ref{lemmadown};
and (b) follows from Cauchy-Schwarz inequality.
Therefore, the left inequality of Theorem \ref{maintheorem} is proved and the proof of Theorem \ref{maintheorem} is complete.

\section{\large Proof of Theorem \ref{algorithm}} \label{theorem2}
Note that RP is a mixed integer linear program. Linear relaxation of RP, denoted by LR-RP,  replaces the constraint $x_{ij}\in\{0,1\}$ with $x_{ij}\geq 0$ for $i\in[1:N],\quad j\in [1:M]$.
Any solution to LR-RP can be denoted by an $N$-by-$M$ matrix $\bold x=[x_{ij}]_{N\times M}$.
So, let $\bold {x^{*}}=[x^{*}_{ij}]_{N\times M}$ denote a basic optimal solution to  LR-RP with objective value $V^*$; i.e., $V^*=\sum_{i=1}^{N}\sum_{j=1}^{M}x^{*}_{ij}$. 
Define 
\begin{align}
& Z_1=\{j\in \{1,2,\ldots ,M \}  | \sum_{i=1}^{N} x^*_{ij}=0\} \nonumber \\
&  Z_2=\{j \in \{1,2,\ldots ,M \} | 0< \sum_{i=1}^{N} x^*_{ij}<1\} \nonumber  \\
& Z_3=\{j \in \{1,2,\ldots ,M \} | \sum_{i=1}^{N} x^*_{ij}=1,  \sum_{i=1}^{N}\lfloor x^*_{ij}\rfloor =0\} \nonumber  \\
& Z_4=\{j \in \{1,2,\ldots ,M \} | \sum_{i=1}^{N} x^*_{ij}=1,  \sum_{i=1}^{N}\lfloor x^*_{ij}\rfloor =1\}. \nonumber  
\end{align}
It is easy to see that $Z_1,Z_2,Z_3,Z_4$ partition the set $\{ 1,2,\ldots , M \}$. Therefore, $M=|Z_1|+|Z_2|+|Z_3|+|Z_4|$.
Furthermore, according to definitions of $V^*,Z_1$, and $Z_4$,
\begin{align}
& C_{det}\leq V^* \leq M-|Z_1| , \label{googooli1}\\
&  \sum_{i=1}^{N}\sum_{j=1}^{M}\lfloor x^*_{ij} \rfloor  =|Z_4|.\label{googooli2}
\end{align}
Hence, by considering (\ref{googooli1}) and (\ref{googooli2}), for proving $C_{det}-  \sum_{i=1}^{N}\sum_{j=1}^{M}\lfloor x^*_{ij} \rfloor \leq N$, it is sufficient to prove 
\begin{equation}\label{magool6}
M-|Z_1|-|Z_4|\leq N,\qquad \text{or equivalently,}\qquad  |Z_2|+|Z_3| \leq N.
\end{equation}
We use a similar approach to \cite{Trick}, \cite{Benders} . 
Note that since $\bold x=[x_{ij}]_{N\times M}$ is a basic solution to LR-RP, the number of inequalities in (\ref{magool1})-(\ref{magool3}) tightened by $\bold x$ is at least the total number of  variables, $MN$. So, if we denote the number of non-tight inequalities in (\ref{magool1}), (\ref{magool2}), (\ref{magool3}) by $n_1,n_2,n_3$, 
\begin{align}
& (N-n_1)+(M-n_2)+(MN-n_3) \geq MN\nonumber\\
&\Rightarrow \qquad n_1+n_2+n_3 \leq M+N. \label{magool5}
\end{align}
On the other hand, according to definition of $Z_1,Z_2,Z_3,Z_4$, we have
\begin{align}
&n_1\geq 0 \\
& n_2\geq |Z_1|+|Z_2|\\
& n_3 \geq |Z_2|+2|Z_3|+|Z_4|, \label{magool4}
\end{align}
where (\ref{magool4}) follows by counting the number of $x^*_{ij}>0$'s  with index $j$  in $Z_2,Z_3$ or $Z_4$; the number of $x^*_{ij}>0$ for which $j\in Z_3$ is at least $2|Z_3|$ since there should be at least two positive fractional values that add up to 1. 
Hence, by (\ref{magool5})-(\ref{magool4}),
\begin{align}
& |Z_1|+2|Z_2|+2 |Z_3|+|Z_4| \leq M+N \quad  \Rightarrow  \quad  |Z_2|+|Z_3| \leq  N, \nonumber
\end{align}
which is the desired inequality as stated in (\ref{magool6}); therefore, the proof is complete.
$\blacksquare$

\begin{corollary}\label{apxresult}
Suppose we choose a basic optimal solution to the LP relaxation of (\ref{relaxed}), denoted by $\bold {x^{*}}$, and round down the solution to get integral values. 
Let $\vec\Pi_{\text{det}}^{\text{apx}}$ denote the assignment suggested by the resulting integral values; and let $\eta_{\text{det}}^{\text{apx}}$ denote the corresponding greedy static scheduling policy.
For $C_{\text{{\sf T}}^{\text{3}}}>\frac{11N}{4}$ we have
\begin{equation}
C_{\text{{\sf T}}^{\text{3}}}-2N-2\sqrt {N(C_{\text{{\sf T}}^{\text{3}}}-\frac{7N}{4})}\leq ||\vec R(\eta_{\text{det}}^{\text{apx}})||_1  \leq C_{\text{{\sf T}}^{\text{3}}}.\nonumber
\end{equation}
\end{corollary}

\begin{proof}
Let $C_{\text{det}}^{\text{apx}}\triangleq \sum_{i=1}^{N}\sum_{j=1}^{M}\lfloor x^{*}_{ij}\rfloor$ denote the  objective value of the rounded down basic optimal solution to  LR-RP.
According to Theorem \ref{maintheorem} and Theorem \ref{algorithm}, $C_{\text{det}}^{\text{apx}}\geq C_{\text{det}}-N\geq C_{\text{{\sf T}}^{\text{3}}}-2N$. 
Therefore, by using the similar argument as  in Corollary \ref{run} the proof will be complete.
\end{proof}

\section{\large Extensions}\label{extensions}
In this section we investigate four important extensions to our main problem formulation:
time-varying channels and real-time traffic; weighted total timely throughput;
lifting the restriction on splitting packets among {\sf AP}'s;
and fading channels, {\sf AP}'s accessing multiple clients simultaneously, clients receiving packets from multiple {\sf AP}'s, and  rate adaptation.
 

\subsection{Time-Varying Channels and Real-Time Traffic}\label{traffic}
So far, we have  assumed that at the beginning of each interval each client has request for exactly one packet. 
This assumption can be modified by considering a time-varying packet generation pattern, in which for every interval, each client might have request for no packets, or  for multiple packets. 
In addition, the number of packets requested by clients for one interval might depend on the number of packets requested for other intervals.
Furthermore, we have so far assumed that  channel success probabilities do not change over time. 
But, this model can be generalized to include time-varying channels with statistical behaviors that are not necessarily independent of one another.

We capture the above two generalizations by considering an irreducible Finite-State Markov Chain (FSMC), in which each state jointly specifies the number of packets requested by each client, as well as the channel states for different channels during an interval. 
When a new interval begins, the Markov Chain might change its state, and in this case, packets for a new subset of clients are requested, and the channel reliabilities change.
Denote the set of all possible states of the FSMC by $\mathcal C$. Each state $\lambda\in\mathcal C$ specifies a pair $(\vec B(\lambda),\bold{P}(\lambda))$, where $\vec B(\lambda)\triangleq [B_1(\lambda),B_2(\lambda),\ldots ,B_M(\lambda)]$, such that $B_j(\lambda)$ is the number of the packets requested by client $j$, and $\bold{P}(\lambda)$ is an $N\times M$ matrix that contains channel success probabilities.
It is assumed that channel success probabilities remain the same during each interval, and are known to the {\sf AP}'s.

Our objective is again to find $C_{\text{{\sf T}}^{\text{3}}}$.
We use a similar argument to the one in \cite{Realtime} for extensions to time-varying channels and variable-bit-rate traffic. 
In particular, 
we decompose the set of intervals into different subsets, where each subset contains the intervals that are in the same state of the FSMC.
For those intervals in which the system is at state $\lambda$, we convert our problem to an instance of the problem described in Section \ref{model}. 
More particularly, for the system described by state $\lambda$, we ignore all the clients that do not have packet request. 
Furthermore, for any $j\in\{1,2,\ldots ,M\}$ where $B_j(\lambda)>1$ we  consider $B_j(\lambda)-1$ virtual clients, such that the channel between $\text {{\sf AP}}_i$ and each of those virtual clients would have success probability $P_{ij}(\lambda)$. This means that these virtual clients are copies of $\text{ {\sf Rx}}_j$. Consequently, for the intervals for which the system is at state $\lambda$ the problem becomes the same as described in Section \ref{model}.
With the same argument as in proof of Theorem \ref{maintheorem}, there exists a fixed assignment $\vec\Pi(\lambda)$, which if used together with its corresponding optimal ordering for such intervals, achieves the optimal $\text{{\sf T}}^{\text{3}}$ for those intervals.
We denote this optimal $\text{{\sf T}}^{\text{3}}$ by $C_{\text{{\sf T}}^{\text{3}}}(\lambda)$.
In addition, let $C_{\text{det}}(\lambda)$ denote the solution to the relaxed problem when the system is at state $\lambda$. 
For any state $\lambda\in \mathcal C$, with the same argument as in the proof of Theorem \ref{maintheorem}, we have
$C_{\text{det}}(\lambda)-2\sqrt {N(C_{\text{det}}(\lambda)+\frac{N}{4})}< C_{\text{{\sf T}}^{\text{3}}}(\lambda)<C_{\text{det}}(\lambda)+N.$
Now, let $\pi_\lambda$ denote the steady state probability of  $\lambda$. Therefore, 
$C_{\text{{\sf T}}^{\text{3}}}=\sum_{\lambda\in \mathcal C} \pi_\lambda C_{\text{{\sf T}}^{\text{3}}}(\lambda),  C_{\text{det}}=\sum_{\lambda\in \mathcal C} \pi_\lambda C_{\text{det}}(\lambda).$  
Hence, 
\begin{equation}
C_{\text{det}}-2\sum_{\lambda\in\mathcal C}\pi_\lambda \sqrt {N(C_{\text{det}}(\lambda)+\frac{N}{4})}< C_{\text{{\sf T}}^{\text{3}}}<C_{\text{det}}+N.\label{aslan}
\end{equation}
On the other hand, by using Cauchy-Schwarz inequality we have
\begin{align}
\sum_{\lambda\in\mathcal C}\pi_\lambda \sqrt {N(C_{\text{det}}(\lambda)+\frac{N}{4})} & \leq \sqrt{\sum_{\lambda\in\mathcal C}\pi_\lambda}\sqrt {\sum_{\lambda\in \mathcal C}N\pi_\lambda(C_{\text{det}}(\lambda)+\frac{N}{4})}=\sqrt {N(C_{\text{det}}+\frac{N}{4})}.\label{gg}
\end{align}
Putting (\ref{aslan}) and (\ref{gg}) together we get
$C_{\text{det}}-2\sqrt {N(C_{\text{det}}+\frac{N}{4})}< C_{\text{{\sf T}}^{\text{3}}}<C_{\text{det}}+N,$
which is the same as the result in Theorem 1.

\begin{theorem}
For the network model described in Section \ref{model} consider the extension to time-varying channels and real-time traffic, modeled by the FSMC described in Section \ref{traffic}, where each state of FSMC captures both the success probability of channels and the number of packets for each client during an interval. We have 
\begin{equation}
C_{\text{det}}-2\sqrt {N(C_{\text{det}}+\frac{N}{4})}< C_{\text{{\sf T}}^{\text{3}}}<C_{\text{det}}+N.\nonumber
\end{equation}
\end{theorem}

\subsection{Weighted Total Timely Throughput}\label{weightedt3}
In Section \ref{model} we considered the same importance for all the flows in the network; and our objective was to maximize $\text{{\sf T}}^{\text{3}}$. However, it might be the case that in a network some of the flows are more important than the others, and should be prioritized accordingly. 
In this section the formulation remains the same as the one described in Section \ref{model}, except the objective function, which rather than maximizing $\text{{\sf T}}^{\text{3}}$, maximizes a weighted average of timely throughputs.
In particular, weighted total timely throughput of the scheduling policy $\eta$, $w\text{-{\sf T}}^{\text{3}}(\eta)$, is defined as
\begin{equation}\label{wt3rj}
w\text{-{\sf T}}^{\text{3}}(\eta)\triangleq \sum_{j=1}^{M}\omega_j R_j(\eta),
\end{equation}
where  $\omega_j$'s are arbitrary weights greater than $1$ $(j=1,2,\ldots, M)$.
Our objective is to find 
\begin{equation}\label{weightedmain}
C_{w\text{-T}^{\text{3}}}\triangleq\sup_{\eta\in\mathcal S}\quad w\text{-{\sf T}}^{\text{3}}(\eta).
\end{equation}

For this extension of the problem we again propose the channel relaxation which results in a new integer program. This integer program is again a GAP.
The formulation of the relaxed problem is as follows:
\begin{eqnarray}
C_{w\text{-det}}\triangleq &&\max \sum_{i=1}^{N}\sum_{j=1}^{M} x_{ij}\omega_{j}\label{weighteddet}\\
&&\text{s.t.}\quad  \sum_{j=1}^{M} \frac{x_{ij}}{p_{ij}}\leq \tau,\quad \sum_{i=1}^{N}x_{ij}\leq 1, \quad\quad x_{ij}\in \{0,1\},\qquad i=[1:N]\quad j=[1:M]\nonumber.
\end{eqnarray}

The following theorem, which is proved in the Appendix \ref{weightedproof}, states that the value of the solution to (\ref{weightedmain}) is asymptotically the same as the value of the solution to (\ref{weighteddet}) as $C_{w\text{-{\sf T}}^{\text{3}}}\to\infty$ (or equivalently $C_{w\text{-det}}\to\infty$).
\begin{theorem}\label{weightedtheorem}
Let $C_{w\text{-{\sf T}}^{\text{3}}}$ denote the value of the solution to (\ref{weightedmain}).
Further, let $C_{w\text{-det}}$ denote the value of the solution to (\ref{weighteddet}). Then,
for $\omega_{max}=\max \{\omega_1,\omega_2,\ldots,\omega_M\}$,
\begin{equation}
C_{w\text{-det}}-2\omega_{max}\sqrt {N(C_{w\text{-det}}+\frac{N}{4})}<C_{w\text{-{\sf T}}^{\text{3}}}<  C_{w\text{-det}}+N\omega_{max}.\label{weightedresult}
\end{equation}
\end{theorem}

\subsection{Dynamic Splitting}\label{optimalonline}
 We  assumed in  Section \ref{model} that the packets are partitioned between {\sf AP}'s at the beginning of each interval,  to reduce the overhead for tracking ACKs and NACKs in the network. If packets are available to all {\sf AP}'s for transmission (i.e., no partitioning is done beforehand), in order to maximize the total timely throughput, each {\sf AP} has to constantly track all ACKs and NACKs of all clients, in order to know whether a  packet has already been delivered to its destination. 
Here we lift the partitioning restriction to understand how much capacity gain can be obtained. 
We first describe the model, and formulate the problem as a Markov Decision Process (MDP).
We then  discuss the tractability of solving the MDP,  propose a  fast greedy heuristic for the MDP, and analyze its computational complexity. 
Finally, we show the performance of the proposed heuristic via numerical results.

\subsubsection{Network Model}
We consider the same network configuration, time model, channel model, and packet arrival as in Section \ref{model}.
Nevertheless, the packets requested for each interval are now available to all {\sf AP}'s (i.e. they are not split among the {\sf AP}'s at the beginning of each interval), and a packet might be served by several {\sf AP}'s.
The {\sf AP}'s can then dynamically choose what packet to transmit in a coordinated manner at each time-slot.
The choice of the packet to be sent by each {\sf AP} may be based on the channels and the past outcomes of the transmissions.
Our objective is to find a scheduling policy which maximizes the total timely throughput of the system, as defined in Section \ref{model}. We call the optimal scheduling the ``Optimal Online Scheduling", since each {\sf AP} has to decide what the optimal strategy is at each time-slot.

\subsubsection{An MDP Formulation}
One can argue in a similar way as in Lemma 1 that due to the time-homogeneous structure of the system, the maximal total timely throughput is equal to the maximum achievable expected number of deliveries in one interval. 
Therefore, the new problem can be formulated as a finite-horizon Markov Decision Process (MDP), as detailed below:

\underline{State Space:} 
The state of the system is an $(M+1)$-tuple where the first $M$ components are binary variables $\{Q_j(t)\}_{j=1}^{M}$, and $Q_j(t)=1$ if $Rx_j$ has not yet received its packet successfully, and $Q_j(t)=0$ otherwise.
The ($M+1$)-th component is the time-slot that the system is currently at, i.e. $Q_{M+1}\in \{1,2,\ldots ,\tau\}$.
We denote the state space by $\mathcal Q$.

\underline{Action Space:}
 For any state $s\in \mathcal Q$ corresponding to set of clients $U(s)$ not having received their packets yet, the action space is an $N$-tuple $(i_1,\ldots, i_N)$ where $i_k\in U(s)\cup \{0\}$ for $k\in \{1,2,\ldots, N\}$.
If $i_k=j$, it means that client $j$ is served by ${\sf AP}_k$, and if $j=0$,  ${\sf AP}_k$ will be idle during the time-slot.
A policy $\mathcal P$ is a function mapping the state space to action space.

\underline{Reward:}
 For successful delivery of each packet, a reward equal to $1$ is obtained.

\underline{Transition Function:}
For $t<\tau$, transition probability from state $s=(q_1,\ldots , q_M,t)$ to state \\$s'=(q'_1,\ldots , q'_M,t+1)$  using action $a(s)$ is simply probability of the event in which in one time-slot using action $a(s)$ the state changes from $s$ to $s'$.
\footnote[1]{More specifically, the transition probability is
$Pr(\cap_{j=1}^{M} (\cup_{1\leq k\leq N, i_k=j}B(p_{kj})=q_j-q'_j)),$
where $B(p)$ is a Bernoulli random variable with parameter $p$.}

 \underline{Objective:} We want to find the optimal policy that maximizes the expected number of deliveries in $\tau$ time-slots.
The objective is similar to the objective initially considered in Section II.

One can use the common technique of using Dynamic Programming to calculate the maximal value.
More specifically, for 2 {\sf AP}'s ($N=2$), the optimization problem reduces to the following.

Let $V^t(U)$ denote the maximum expected number of deliveries for the set of packets $U$ and during time-slots $t,t+1,\ldots ,\tau$.
Therefore, the objective can be rewritten as follows.
\begin{align}
\text{Objective:} \qquad    V^{1}(\{1&,2,\ldots, M\}),\nonumber\\
  \text{where}  \qquad V^t(U) &=\max_{\{i,j\}\in U}\{ p_{1i}p_{2j}[2-\mathbb I (i=j)+V^{t+1}(U\setminus\{i,j\})]\nonumber\\
&+p_{1i}(1-p_{2j})[1+V^{t+1}(U\setminus\{i\})]\nonumber\\
\qquad\qquad\qquad\qquad   &+(1-p_{1i})p_{2j}[1+V^{t+1}(U\setminus\{j\})]+(1-p_{1i})(1-p_{2j})V^{t+1}(U)\},\nonumber
\end{align}
and $V^{\tau}(U)=\max_{\{i,j\}\in U} [p_{1i}+p_{2j}-p_{1i}p_{2j}\mathbb I(i=j)]$, where $\mathbb I(.)$ is the indicator function.

Computational complexity of solving the DP is polynomial in $\tau$, but exponential in $M$.
This complexity grows even faster as $N>2$.
Hence, calculating the optimal solution is challenging. However, we will propose a fast greedy heuristic that  approximates the optimal solution well.

\subsubsection{A Greedy Heuristic}

The greedy heuristic (Algorithm \ref{greed}) essentially ignores time, and at each time-slot sends a subset of packets by the {\sf AP}'s which would maximize the expected number of deliveries for that specific time-slot. 
Moreover, according to Lemma \ref{greedlemma}, for finding the subset of packets which results in the maximum expected delivery for a time-slot, it is not necessary to search over all $N^M$ possible subsets; instead, it is sufficient to only focus on  $N^N$  subset of them.
Algorithm \ref{greed} is repeated for all intervals.

\begin{algorithm}
\caption{}
\begin{algorithmic}\label{greed} 

\STATE Set $t=1$ and $U=\{1,2,\ldots ,M\}$. 
\STATE Create $N$ vectors $L_1,\ldots ,L_N$, and put the packets $\{1,2,\ldots ,M\}$ in all of them.
\STATE Order packets in each $L_k$, $k\in\{1,2,\ldots ,N\}$, according to $p_{kj}$'s and in decreasing order.
\WHILE {$t\leq \tau , U\ne \Phi$}
\STATE Find $[j_1,\ldots ,j_N]=\arg\max_{j_1\in L_1(1:N),\ldots, j_N\in L_N(1:N)} \{\sum_{i=1}^{M}(1-\Pi_{j_m=i, 1\leq m\leq N}(1-p_{mi}))\}.$
\STATE Transmit $j_1,j_2,\ldots, j_N$ by ${\sf AP}_1,{\sf AP}_2,\ldots, {\sf AP}_N$ respectively.
\STATE Update $L_1,L_2,\ldots ,L_N$ according to the outcome of transmissions (remove any of $j_1,j_2,\ldots, j_N$ from them which is successfully delivered, and shift the queues to the left to fill the gap of the removed packets). Also, remove the delivered packets from $U$.
\STATE $t\leftarrow t+1$
\ENDWHILE

\end{algorithmic}
\end{algorithm}

In fact, $\sum_{i=1}^{M}(1-\Pi_{j_m=i, 1\leq m\leq N}(1-p_{mi}))\}$ is the expected number of deliveries for a time-slot, when $j_m$ is transmitted by ${\sf AP}_m$.
The following lemma establishes why if packets are ordered in the queues of {\sf AP}'s, then for finding the subset of packets which results in maximum expected delivery for the time-slot, it is sufficient to just look at the first $N$ elements of each queue.

\begin{lemma}\label{greedlemma}
Suppose $[j_1:j_N]=\arg\max_{j_1\in L_1(1:N),\ldots, j_N\in L_1(1:N)} \{\sum_{i=1}^{M}(1-\Pi_{j_m=i, 1\leq m\leq N}(1-p_{mi}))\}.$
If $1\geq p_{kL_k(1)}\geq p_{kL_k(2)}\geq \ldots, \geq p_{kL_k(|U|)}\geq 0$ for $k\in\{1,2,\ldots ,N\}$,
then
$$\sum_{i=1}^{M}(1-\Pi_{j_m=i, 1\leq m\leq N}(1-p_{mi}))
=\max_{j'_1,\ldots, j'_N\in \{1,2,\ldots, M\}} \{\sum_{i=1}^{M}(1-\Pi_{j'_m=i, 1\leq m\leq N}(1-p_{mi}))\}.$$
\end{lemma}

\begin{proof}
Consider the $N$ vectors $L_1,\ldots ,L_N$, defined in Algorithm \ref{greed}, where  packets in each $L_k$, $k\in\{1,2,\ldots ,N\}$ are ordered according to $p_{kj}$'s and in decreasing order, meaning that $1\geq p_{kL_k(1)}\geq p_{kL_k(2)}\geq \ldots, \geq p_{kL_k(M)}\geq 0$ for $k\in\{1,2,\ldots ,N\}$.
Suppose there is no subset of packets $j_1,j_2,\ldots, j_N$ such that each $j_k$ is one of the first $N$ elements of $L_k$, and 
$j_1,j_2,\ldots, j_N$ maximizes the expected deliveries over a time-slot for the set of packets $\{1,2,\ldots, M\}$.
More precisely, suppose there is no $j_1,j_2,\ldots, j_N$  such that $j_k\in L_k(1:N)$ for $k\in\{1,2,\ldots, N\}$, and it maximizes the $\sum_{i=1}^{M}(1-\Pi_{j_m=i, 1\leq m\leq N}(1-p_{mi}))\}$. 
Consider an arbitrary $j_1,j_2,\ldots, j_N$ which maximizes the expected delivery $\sum_{i=1}^{M}(1-\Pi_{j_m=i, 1\leq m\leq N}(1-p_{mi}))\}$.
Therefore, there is one of the $j_k$'s that does not belong to the first $N$ elements of $L_k$.
More precisely, there exists a $k$, $k\in \{1,2,\ldots ,N\}$, for which $j_k\notin L_k(1:N)$. 
Therefore, there is at least one of the first $N$ elements of $L_k$ which is not going to be transmitted by any {\sf AP}.
In other words, there must exist an $l$ such that $l\in L_k(1:N)$, and $l\ne j_i$ for $i\in\{1,2,\ldots ,N\}$.
Since, $p_{kl}> p_{kj_k}$, by serving $l$ on ${\sf AP}_k$ the expected deliveries, $\sum_{i=1}^{M}(1-\Pi_{j_m=i, 1\leq m\leq N}(1-p_{mi}))\}$, will only increase. This contradicts the assumption that $j_1,j_2,\ldots, j_N$ produce the maximal expected number of deliveries; 
and therefore, $j_k\in L_k(1:N)$ for $k\in \{1,2,\ldots, N\}$, and the proof is complete.
Note that although the lemma and its proof are stated for the set of packets $\{1,2,\ldots, M\}$, they hold for any arbitrary set of packets $U$, too.
\end{proof}

The total processing time of Algorithm \ref{greed} is $O(\tau MN^{N+1})$; since
the while loop is run $\tau$ times, and finding $j_1,j_2,\ldots ,j_N$ takes $O(N^NMN)=O(MN^{N+1})$.

\subsubsection{Numerical Results}

We compare the total timely throughput capacity for optimal online policies, splitting policies ($C_{{\sf T}^3}$), and greedy heuristic (Algorithm \ref{greed}). 

Heuristic Algorithm \ref{greed} is not optimal in general. 
However, as numerical results  indicate, the value provided by the greedy algorithm is quite close to the optimal value. In fact, the numerical results suggest that  Algorithm \ref{greed} is a decent approximation for the optimal value.

\begin{figure}
\centering
\subfigure[]{\label{configur}\includegraphics[scale=.4,trim = 50mm 30mm 40mm 60mm]{simulation121011networkconfig.pdf}}
\subfigure[]{\label{greedyfig}\includegraphics[scale=.5, trim= 30mm 80mm 30mm 90mm]{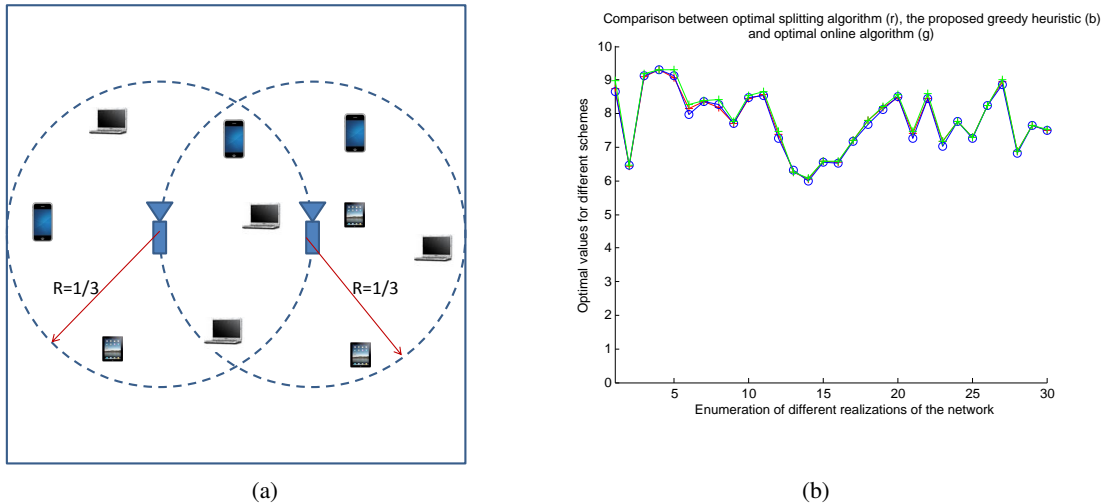}}
\caption{(a) wireless network with two {\sf AP}'s with coverage radius $\frac{1}{3}$, $10$ randomly and uniformly located clients in the coverage area of the {\sf AP}'s with channel erasure probabilities proportional to the distance, and $\tau=15$,  for 30 different realizations of the network. 
In (b) the red curve demonstrates the total timely throughput capacity when the scheduling is restricted to partitioning  the set of packets across {\sf AP}'s. The green curve demonstrates the total timely throughput capacity when the splitting assumption is relaxed. The blue curve demonstrates the total timely throughput achieved by the greedy heuristic described in Algorithm \ref{greed}. The curves demonstrate that (i) the greedy heuristic solution is near optimal; and (ii)
 very marginal capacity gain can be obtained by relaxing the partitioning assumption.}
\end{figure}

Interestingly, as numerical results in Fig. \ref{greedyfig} indicate, the throughput of offline splitting algorithm is very close to that of optimal online scheduling which is the maximum throughput over all possible policies. 
Hence, lifting the assumption of partitioning traffic among {\sf AP}'s  provides marginal gain over the optimal splitting algorithm, while it requires much more coordination of ACK/NACKs.
Consequently, for a system-level design, one may only focus on how to split the traffic among different {\sf AP}'s, and they are ensured that the  solution will be near optimal.

\subsection{Fading Channels and Rate Adaptation}\label{rateadaptation}
Section \ref{model}  considered a packet erasure model for channels, and assumed  that each {\sf AP} can transmit one packet at a time.
We extend the model to  consider fading channels in order to better capture the channel physical properties. 
In addition, we allow  each {\sf AP} to allocate a portion of its available bandwidth to each client during a time-slot.
This means that each {\sf AP} can access several clients simultaneously.
Moreover, we allow for rate adaptation, where according to the time-frequency resource allocated to each client, a certain \emph{reward} will be obtained.

\subsubsection{Model Setup}
Consider the network topology and time model  described in Section  \ref{model}.
In addition, for $i\in\{1,2,\ldots, N\}$, ${\sf AP}_i$ has bandwidth $W_i$, where $W_i\in \mathbb N$, which means at most $W_i$ simultaneous transmissions can occur during a time-slot by ${\sf AP}_i$.
On the other hand, all the bandwidth of ${\sf AP}_i$ during a time-slot might be allocated to a certain client.

Define $R_j^{i_1,i_2,\ldots ,i_N}$ to be the total reward obtained by$\text{ {\sf Rx}}_j$ during an interval if it is served $i_1,i_2,\ldots, i_N$ times on ${\sf AP}_1,{\sf AP}_2,\ldots, {\sf AP}_N$, respectively.
The amount of this reward is determined by the rate adaptation which is used in the {\sf AP}'s.
Further, assume that $R_j^{i_1,i_2,\ldots ,i_N}$ for $j=1,2,\ldots , M$ is a non-negative, increasing function in all dimensions $i_1,i_2,\ldots, i_N$.

A scheduling policy $\eta$ for the system allocates, possibly at random, the bandwidth of each {\sf AP} to different clients in each time-slot, based on the past history of the system. 
Let $q_j(k)$ denote the reward obtained for client $\text{ {\sf Rx}}_j$ during interval $k$ under some scheduling policy.
 The average reward for $\text{ {\sf Rx}}_j$ is defined as $q_j=\limsup_{k\to\infty}\frac{\sum_{i=1}^{k}q_j(i)}{k}$.
The objective is to maximize $\sum_{j=1}^{M}q_j$, which is the total average reward.

\begin{remark}

The Relaxed Problem introduced in Section \ref{results} was in fact a deterministic scheduling problem with binary rewards; i.e. either size $1/p_{ij}$ would be allocated to packet of client $\text{ {\sf Rx}}_j$ in bin $i$, which would result in reward one (it will contribute to the objective function by setting $x_{ij}=1$); or, it would not add to the value of the objective function at all (for $x_{ij}=0$).
Therefore, the value of $\sum_{i=1}^{N}\sum_{j=1}^{M}x_{ij}$ can be viewed as the reward resulting from a scheduling policy.
Nevertheless, a more practical model for the reward is a function with input argument being the amount of time-frequency allocated to the client.
Therefore, the model extension we are considering can also be viewed as a generalization of the deterministic scheduling (RP).  
\end{remark}

A similar model has been considered in \cite{HetReward} for $N=1$, where no simultaneous transmissions are allowed, i.e. the bandwidth of {\sf AP} is equal to $1$, and intervals for clients are not necessarily equal.
They show that for checking if a set of reward requirements is feasible, it is sufficient to look at the average behavior of the system.
 However, when going from one {\sf AP} to multiple {\sf AP}'s checking the average behavior is not sufficient, even when multiple simultaneous transmissions is not allowed, and all deadlines are equal.
We focus on maximizing the total average reward, which is the equivalent of $C_{\text{{\sf T}}^{\text{3}}}$ in our original result.
To this aim, we first state the following lemma which reduces the problem to a maximization problem over an interval of length $\tau$.
Then, we show that this new maximization problem can be solved using a Dynamic Programming.

\begin{lemma}\label{Reduction4Rewards}
The optimization problem can be formulated as follows.
\begin{equation} \label{rewardreduction}
\max\quad \sum_{j=1}^{M}R_j^{x_{1j},\ldots ,x_{Nj}}\quad s.t.\quad \sum_{j=1}^{M}x_{ij}\leq W_i\tau,  \quad  x_{ij}\in \mathbb Z^+\cup \{0\},  \quad i\in [1:N], j\in[1:M].   
\end{equation}
\end{lemma}
\begin{proof}
It is  sufficient to show that the maximal total average reward is obtained using a policy which is implemented for all intervals; since (\ref{rewardreduction}) finds  the maximal total reward over one interval.
The proof in essence is similar to that of Lemma 1.
Consider  the following two observations.
First, we have a finite number of possible actions to take for each interval. More specifically, since we have $M$ clients, $N$ {\sf AP}'s, and $W_i\tau$ chunks of resource in ${\sf AP}_i$, total number of different possible actions for an interval is  $M^{\sum_{i=1}^{N}\tau W_i}$.
Second,  each policy produces a certain reward. Among all possible policies for one interval, there is one policy $\mathcal P$ with maximal total reward $R^*$. 
Hence, any sequence of policies that is implemented on the sequence of intervals produces at most the total average reward of $R^*$, which is obtained by applying $\mathcal P$ to all intervals.
\end{proof}

\subsubsection{Dynamic Programming Solution}
In this part we use Lemma \ref{Reduction4Rewards} to propose a  DP solution to the problem.
Define $OPT[m,t_1,\ldots, t_N]$  to be the maximal total reward obtained when only scheduling the first $m$ clients, with the available resource  being $t_1,t_2,\ldots, t_N$ on ${\sf AP}_1,{\sf AP}_2,\ldots, {\sf AP}_N$, respectively.
Hence, our objective is to find $OPT[M,W_1\tau,W_2\tau,\ldots, W_N\tau]$.

\begin{algorithm}
\caption{}
\begin{algorithmic}\label{DP} 

\STATE Input $R_j^{i_1,i_2,\ldots, i_N}$ for $1\leq j\leq M, 0\leq i_1\leq \tau W_1,0\leq i_2\leq \tau W_2,\ldots, 0\leq i_N\leq \tau W_N.$
\STATE  Initialize a $[M\times (W_1\tau+1)\times \ldots \times (W_N\tau+1)]$ array $OPT$.
\FOR {$m=1,2,\ldots, M$}
\FOR {$t_1\in [0:\tau W_1],\ldots ,t_N\in [0:\tau W_N]$}
\IF {$m=1$}
\STATE $OPT[m,t_1,\ldots, t_N]\leftarrow R_1^{t_1,\ldots, t_N}$;
\ELSE
\STATE $OPT[m,t_1,\ldots ,t_N]  \leftarrow       \max_{0\leq x_1\leq t_1,\ldots,0\leq x_N\leq t_N} 
\{OPT[m-1,t_1-x_1,\ldots ,t_N-x_N]+R_m^{x_1,\ldots, x_N}\}$;
\ENDIF

\ENDFOR
\ENDFOR
\STATE Output $OPT[M,W_1\tau,\ldots, W_N\tau]$.
\end{algorithmic}
\end{algorithm}

\begin{theorem}\label{DynProg}
Algorithm \ref{DP}  solves the problem of finding the maximum total average reward  in $O(M\tau^{2N}\Pi_{i=1}^{N}W_i^2)$ time.
\footnote[1]{The same methodology of applying Dynamic Programming can be used to solve the problem when packets arrive at the beginning of intervals, but they have different deadlines during the interval.}
\end{theorem}

\begin{proof} The proof contains two parts: processing time of the algorithm, and proof of correctness.
Total Processing Time:
there are total of $O(M\tau^N\Pi_{i=1}^{N}W_i)$ iterations, each with computational complexity of $O(\tau^N\Pi_{i=1}^{N}W_i)$.
Therefore, the total processing time is $O(M\tau^{2N}\Pi_{i=1}^{N}W_i^2)$, which is polynomial in the number of clients.
(Also, note that the number of {\sf AP}'s, $N$, is typically small, around 2,3, or 4.)
Proof of Correctness:
The algorithm stores an $(N+1)$-dimensional array $OPT$. 
 We use induction  over the entries of the dynamic programming table, in order that the algorithm fills them in.
Induction hypothesis is that $OPT[m,t_1,\ldots, t_N]$ is the maximal total reward when there are only the first $m$ clients in the system, and the available resource on ${\sf AP}_1,\ldots, {\sf AP}_N$ are $t_1,\ldots, t_N$ respectively. 
For the base case of $m=1$ the algorithm allocates all the available resource to the first client, and the table is initialized correctly. We now check the induction step. Consider the time when $OPT[m,t_1,\ldots, t_N]$ is going to be computed by the algorithm; and assume  all the previous entries of the table $OPT$ have been correctly computed.
First, note that all the entries of the table that the recursive formula for finding $OPT[m,t_1,\ldots, t_N]$   is referring to have already been computed in  earlier steps.
Second, note that the maximization  in the recursive formula accounts for all the possible allocations of the resource to the $m$-th client, and then for each allocation it computes the maximal total reward, which is the reward using that allocation for client $m$ plus the  maximal reward for the subproblem of only having the first $m-1$ clients, which is already computed correctly according to the induction hypothesis.  
\end{proof}

\section{\large Numerical Analysis}\label{numerical}
In this section we provide numerical analyses for our deterministic relaxation scheme. 
So, we consider a wireless network with 2 {\sf AP}'s, and several wireless clients that are uniformly and randomly located in the network (see Figure~\ref{configur}). 
Channel from every {\sf AP} to every client is an erasure channel with erasure probability which is proportional to the distance between the {\sf AP} and the client.
The distances in the network are normalized, and we assume that the {\sf AP}'s have the same coverage radius $R=\frac{1}{3}$.
Therefore, the channel erasure probability is 1 for the channel between an {\sf AP} and a client which is located at the distance $R\geq \frac{1}{3}$ from it. Furthermore, the distance between the two AP's is $\frac{1}{3}$.

\begin{figure}
\centering
\subfigure[]{\label{configur}\includegraphics[scale=.4,trim = 30mm 30mm 30mm 110mm]{simulation121011networkconfig.pdf}}
\subfigure[]{\label{121311}\includegraphics[scale=.5,trim = 30mm 80mm 30mm 110mm]{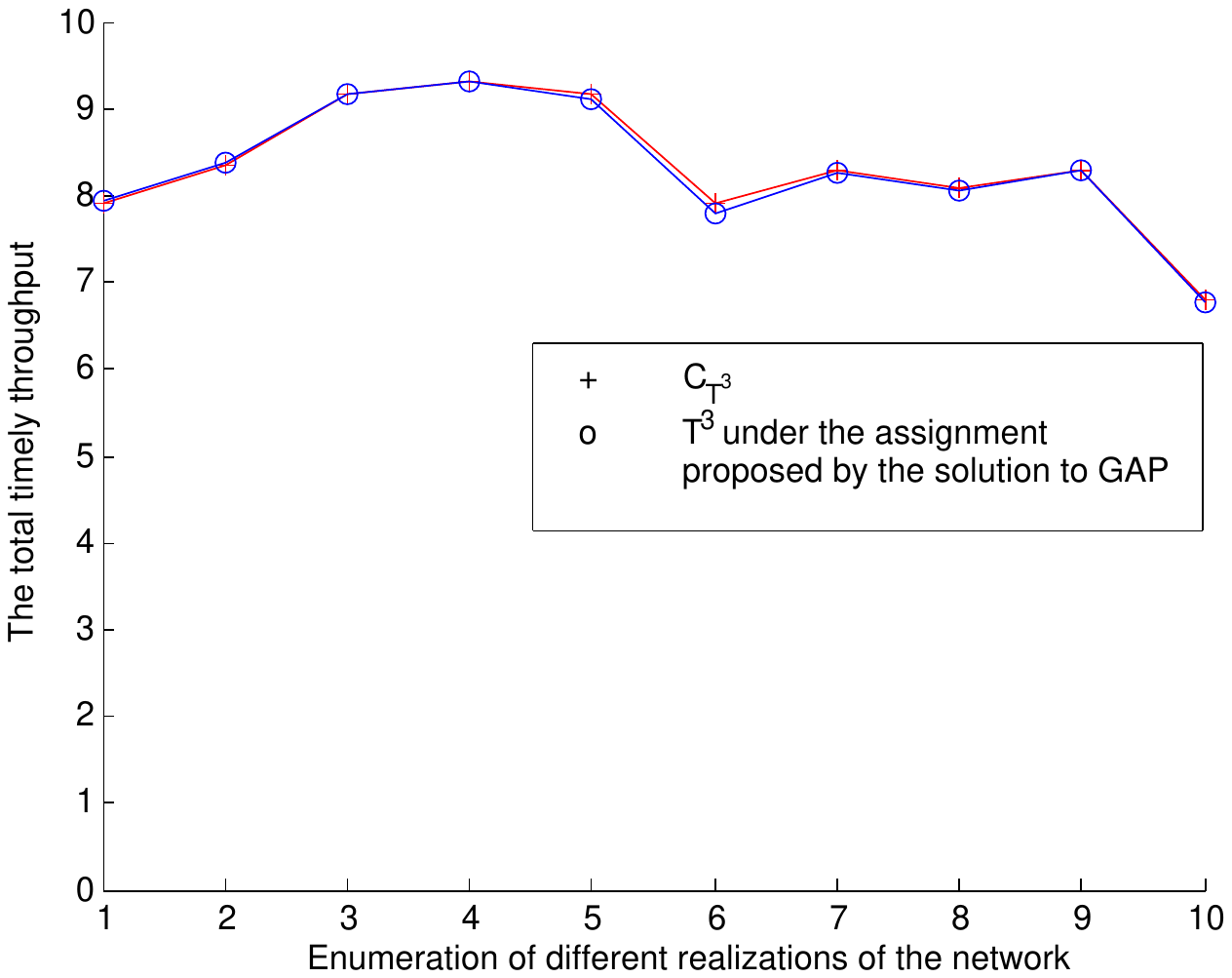}}\\
\subfigure[]{\label{121411}\includegraphics[scale=.5,trim = 30mm 80mm 30mm 80mm]{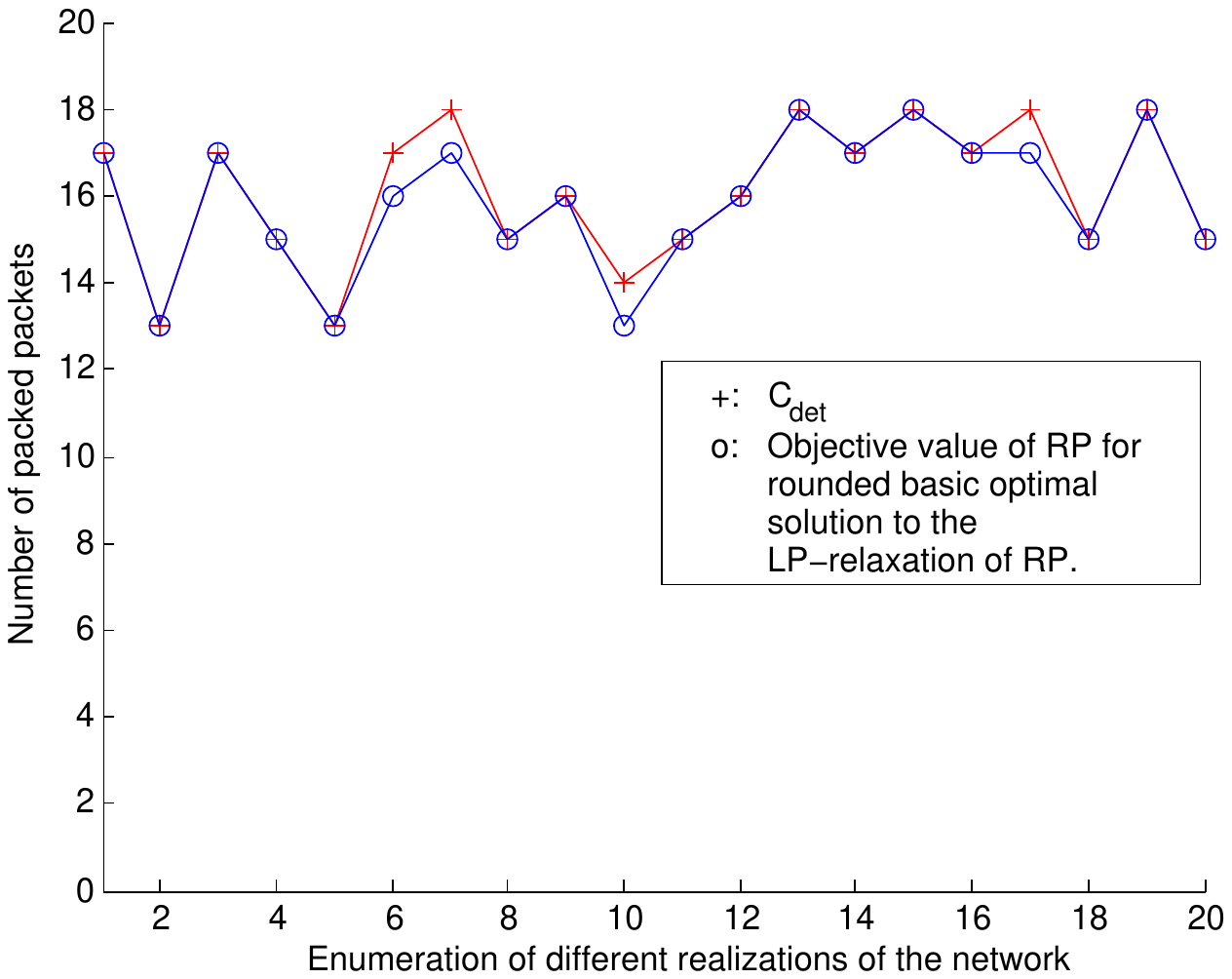}}
\subfigure[]{\label{121511}\includegraphics[scale=.5,trim = 30mm 80mm 30mm 80mm]{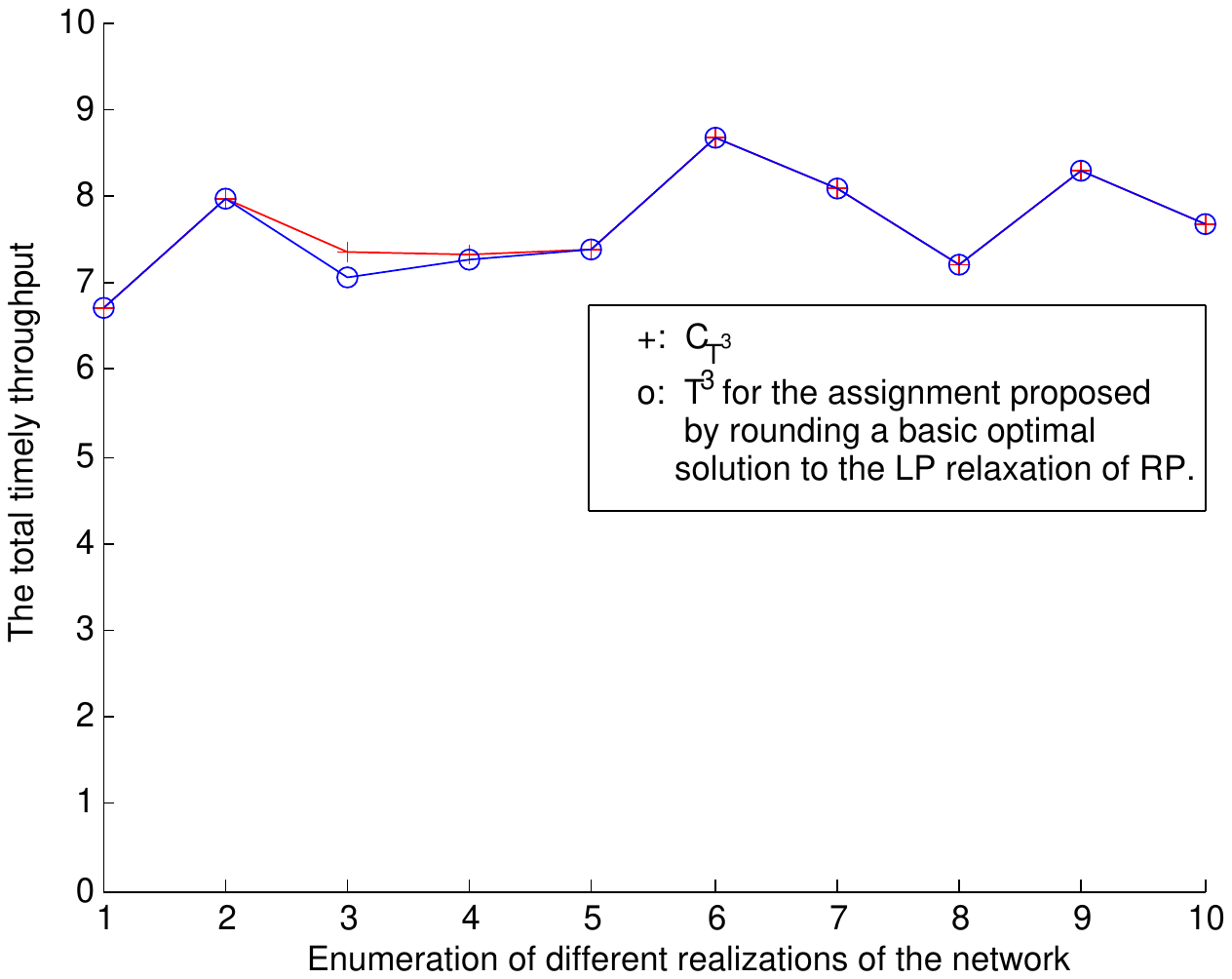}}
\caption{Numerical Results. (a) illustrates the network configuration  with two {\sf AP}'s with coverage radius $\frac{1}{3}$ each, $M$ randomly and uniformly located clients in the coverage area of the {\sf AP}'s, and channel erasure probabilities proportional to the distances.
(b) compares $C_{\text{{\sf T}}^{\text{3}}}$ with the $\text{{\sf T}}^{\text{3}}$  resulted from the assignment proposed by (\ref{relaxed}), $\eta_{\text{det}}$, for  $M=10,\tau=15$ and $10$ different realizations of the network.
`+' and `o' in (b) indicate the values of $C_{\text{{\sf T}}^{\text{3}}}$ and $\text{{\sf T}}^{\text{3}}(\eta_{\text{det}})$ for each realization, respectively.
(c) compares $C_{\text{det}}$ (denoted by `+') with the objective value of the rounded basic optimal solution (denoted by `o') for  $M=20,\tau=30$ and $20$ different network realizations.
Finally, (d) compares $C_{\text{{\sf T}}^{\text{3}}}$ (denoted by `+') with the $\text{{\sf T}}^{\text{3}}$ resulted from the assignment proposed by Algorithm \ref{apxalgorithm} (denoted by `o') for $M=10,\tau=15$ and $10$ different realizations of the network.
}
\end{figure}

Figure ~\ref{121311} corresponds to the case where $M=10$ and $\tau=15$.  In each realization $10$ clients are randomly located in the network. For each realization $C_{\text{{\sf T}}^{\text{3}}}$ is calculated. Then, the corresponding relaxed problem is solved, and the network is run for 10000 intervals under the assignments proposed by its deterministic relaxation solution. Fig \ref{121311} shows the comparison between the two for 30 different realizations of the network.

Figure ~\ref{121411} demonstrates how our proposed LP-rounding algorithm (Algorithm \ref{apxalgorithm}) performs compared to $C_{\text{det}}$.
We consider $M=20$ and $\tau=30$, and 30 different realizations of network. For each realization $C_{\text{det}}$, and the value proposed by our approximation algorithm (Algorithm \ref{apxalgorithm}) are found.
 The result confirms the fact that our proposed algorithm performs well in approximating the optimal solution. The performance improves as the number of clients increases.

Figure ~\ref{121511} shows how far our $\text{{\sf T}}^{\text{3}}$ will be from $C_{\text{{\sf T}}^{\text{3}}}$  if we use Algorithm \ref{apxalgorithm} as the assignment strategy for the packets, and run the network for 10000 intervals according to that assignment. In this case we have considered $M=10,\tau=15,$ and  $10$ different instances of the network.

\section{\large Conclusion}\label{conclusion}
In this work we investigated the improvement by utilizing network heterogeneity in order to enhance the timely throughput of a wireless network.
In particular, we studied the problem of maximizing total timely throughput of the downlink of a wireless network with $N$ Access points and $M$ clients, where each client might have access to several Access points.
This problem is challenging to attack directly. 
However, we proposed a deterministic relaxation of the problem which is based on converting the problem to a network with deterministic delay for each link.

First, we showed that the value of the solution to the relaxed problem, $C_{\text{det}}$, is very close to the value of the solution to the original problem, $C_{\text{\sf T}^{\text{3}}}$. In fact, as $C_{\text{\sf T}^{\text{3}}}\to\infty,\frac{C_{\text{det}}}{C_{\text{\sf T}^{\text{3}}}}\to 1$. 
Furthermore, the numerical results indicate that for networks with limited number of clients, the gap between $C_{\text{\sf T}^{\text{3}}}$ and $C_{\text{det}}$ is very small. Second, we proposed a simple polynomial-time algorithm with additive performance guarantee of $N$ for approximating the relaxed problem. This approximation performs well as $N$ is for most cases between 2-4.
We also extended the formulation to allow time-varying channels, real-time traffic, and weighted total timely throughput maximization, and proved similar results. 
In addition, we extended the model to account for fading, multiple simultaneous transmissions by Access Points, and rate adaptation.
Two future directions are considering  multi-hop model, and allowing different deadlines for clients.

\appendices
\section{Proof of the tightness of the bounds in Theorem 1}\label{tightnessproof}
We prove that the upper and lower bounds given in (\ref{result}) are tight. More specifically, we show that there exist $N,M,$ and some channel success probabilities for which $C_{\text{{\sf T}}^{\text{3}}}$ gets arbitrarily close to $C_{\text{det}}+N$. In addition, there exist $N,M,$ and some channel success probabilities for which $O(|C_{\text{{\sf T}}^{\text{3}}}-C_{\text{det}}|)=O(\sqrt{NC_{\text{det}}})$.

\subsection{Proof of the tightness of the upper bound}

We show that for any given $N$ and $0<\epsilon <N$  there exist $M,\tau$, and channel success probabilities such that $C_{\text{{\sf T}}^{\text{3}}}-C_{\text{det}}=N-\epsilon$.
We set $M=N\tau$, and we choose $C_{\text{det}}$ such that $C_{\text{det}}<M-N$ and $\frac{C_{\text{det}}}{N}\in \mathbb N$. 
Further, for the channel between  $\text {{\sf AP}}_i$ and $\text{{\sf Rx}}_j$ we set the channel success probability $p_{ij}=\frac{C_{\text{det}}+N-\epsilon}{N\tau}$, where $i=1,2,\ldots ,N $ and $j=1,2,\ldots ,M$. 
Therefore, according to symmetry, both the optimal assignment which results in $C_{\text{{\sf T}}^{\text{3}}}$ and the optimal assignment for the relaxed problem which results in $C_{\text{det}}$ assign $\tau$ packets to each {\sf AP}. Furthermore, without loss of generality we can assume that for $\text {{\sf AP}}_i$ packets of clients $j=1+(i-1)\tau,\ldots ,i\tau$ are assigned to  $\text {{\sf AP}}_i$. 
 It is easy to check that the following inequalities hold for any $\text {{\sf AP}}_i$, $i=1,2,\ldots ,N$:
\begin{equation}
\sum_{j=1+(i-1)\tau}^{1+(i-1)\tau+C_{\text{det}}/N}\frac{1}{p_{ij}}=(\frac{C_{\text{det}}}{N})(\frac{N\tau}{C_{\text{det}}+N-\epsilon})<\tau<\sum_{j=1+(i-1)\tau}^{1+(i-1)\tau+C_{\text{det}}/N+1}\frac{1}{p_{ij}}.\nonumber
\end{equation}
Therefore, the maximum number of packets that can be packed in the relaxed problem is $C_{\text{det}}$.
Now, we calculate the expected number of packet deliveries: 
For any $\text {{\sf AP}}_i$ the expected number of successful deliveries during one interval is $\tau(\frac{C_{\text{det}}+N-\epsilon}{N\tau})=\frac{C_{\text{det}}+N-\epsilon}{N}$. Therefore, we have $C_{\text{{\sf T}}^{\text{3}}}=N(\frac{C_{\text{det}}+N-\epsilon}{N})=C_{\text{det}}+N-\epsilon$.
Hence, $C_{\text{{\sf T}}^{\text{3}}}-C_{\text{det}}=N-\epsilon$.

\subsection{Proof of the tightness of the order of the lower bound}
We show that there exists a wireless network realization for which $O(|C_{\text{{\sf T}}^{\text{3}}}-C_{\text{det}}|)=O(\sqrt{NC_{\text{det}}})$.
More specifically, for a given $N$ we show that 
there exist a positive constant $k$ along with $M,\tau$, such that $C_{\text{det}}-C_{\text{{\sf T}}^{\text{3}}}>k\sqrt{NC_{\text{det}}}$.
We choose $C_{\text{det}}$ such that  $\frac{C_{\text{det}}}{N}\in \mathbb N$, and we set $M=C_{\text{det}}$. 
In addition, we set the channel success probability $p_{ij}=p=\frac{C_{\text{det}}}{N\tau}<1$ for some $\tau\in \mathbb N$, where $i=1,2,\ldots ,N $ and $j=1,2,\ldots ,M$. 

Therefore, both the optimal assignment which results in $C_{\text{{\sf T}}^{\text{3}}}$ and the optimal assignment for the relaxed problem which results in $C_{\text{det}}$ assign $\frac{C_{\text{det}}}{N}$ packets to each {\sf AP}. 
It is easy to check that our chosen $C_{\text{det}}$ is actually the solution to the relaxed problem. 
Now, let $Y$ denote the number of successful deliveries for one of the {\sf AP}'s. 
Thus, $C_{\text{{\sf T}}^{\text{3}}}=NE[Y]$.
Also, let $l$ denote the number of packets that can be packed in a bin corresponding to a certain {\sf AP}.
 Therefore, $l=\frac{C_{\text{det}}}{N}$, and $p=\frac{l}{\tau}$. We only need to show that there exists a constant $k$ such that
$l-E[Y]>k\sqrt{l}.$
Noting that $l=p\tau$ we have
\begin{eqnarray}
l-E[Y]&=& p\tau-[\sum_{j=1}^l j\binom{\tau}{j}p^j(1-p)^{\tau-j}+l\sum_{j=l+1}^{\tau} \binom{\tau}{j} p^j(1-p)^{\tau-j}]\nonumber\\
&=& \sum_{j=1}^\tau j\binom{\tau}{j}p^j(1-p)^{\tau-j}-[\sum_{j=1}^lj\binom{\tau}{j}p^j(1-p)^{\tau-j}+l\sum_{j=l+1}^{\tau} \binom{\tau}{j} p^j(1-p)^{\tau-j}]\nonumber\\
&=& p\tau\sum_{j=l+1}^{\tau} \binom{\tau-1}{j-1}p^{j-1}(1-p)^{\tau-j}-l\sum_{j=l+1}^{\tau} \binom{\tau}{j} p^j(1-p)^{\tau-j}\nonumber
\end{eqnarray}
\begin{eqnarray}
&=& l[\sum_{j=l+1}^{\tau} \binom{\tau-1}{j-1}p^{j-1}(1-p)^{\tau-j}-\sum_{j=l+1}^{\tau}( \binom{\tau-1}{j}+\binom{\tau-1}{j-1}) p^j(1-p)^{\tau-j}]\nonumber\\
&=&l[\sum_{j=l+1}^{\tau}\binom{\tau-1}{j-1}p^{j-1}(1-p)^{\tau-j+1}-\sum_{j=l+1}^{\tau-1}\binom{\tau-1}{j}p^{j}(1-p)^{\tau-j}]\nonumber\\
&=&l\binom{\tau-1}{l}p^l(1-p)^{\tau-l}\nonumber.
\end{eqnarray}
Now note that $\binom{\tau-1}{l}=\frac{(\tau-1)!}{l!(\tau-1-l)!}=\frac{\tau-l}{\tau}\binom{\tau}{l}=(1-p)\binom{\tau}{l}$.
Therefore, $l\binom{\tau-1}{l}p^l(1-p)^{\tau-l}=\tau\binom {\tau}{l}p^{l+1}(1-p)^{\tau-l+1}.$
By Theorem 2.6 of \cite{Stanica} we know that for positive integers m,n,q, with $m>q\geq 1$ and $n\geq 1$
\begin{equation}
\binom{mn}{qn}>\frac{1}{\sqrt {2\pi}}   e^{\frac{1}{12n}(\frac{1}{m}-\frac{1}{q}-\frac{1}{m-q})}
n^{-\frac{1}{2}}\frac{m^{mn+\frac{1}{2}}}{(m-q)^{(m-q)n+\frac{1}{2}}q^{qn+\frac{1}{2}}}.\nonumber
\end{equation}
Substituting $n$ by $1$, $m$ by $\tau$, and $q$ by $l$ we get:
\begin{eqnarray}
\binom{\tau}{l}&>&\frac{1}{\sqrt {2\pi}}\frac{\tau^{\tau+\frac{1}{2}}}{(\tau-l)^{(\tau-l)+\frac{1}{2}}l^{l+\frac{1}{2}}}e^{\frac{1}{12}(\frac{1}{\tau}-\frac{1}{l}-\frac{1}{\tau-l})}\nonumber\\
&=&\frac{1}{\sqrt {2\pi}}\frac{\tau^{\tau+\frac{1}{2}}}{(\tau(1-p))^{\tau-l+\frac{1}{2}}(p\tau)^{l+\frac{1}{2}}}e^{\frac{1}{12}(\frac{1}{\tau}-\frac{1}{l}-\frac{1}{\tau-l})}\nonumber\\
&=&\frac{1}{\sqrt {2\pi}}\frac{1}{\sqrt {\tau p(1-p)}}\frac{1}{p^{l}(1-p)^{\tau-l}}e^{\frac{1}{12}(\frac{1}{\tau}-\frac{1}{l}-\frac{1}{\tau-l})}\nonumber.
\end{eqnarray}
However, $\frac{1}{\tau}-\frac{1}{l}-\frac{1}{\tau-l}>-2$. Therefore,
$\binom{\tau}{l}>\frac{1}{\sqrt {2\pi}}\frac{1}{\sqrt {\tau p(1-p)}}\frac{1}{p^{l}(1-p)^{\tau-l}}e^{-\frac{1}{6}}.$
Hence,  we get
\begin{eqnarray}
l-E[Y]&=&l\binom{\tau-1}{l}p^l(1-p)^{\tau-l}=\tau\binom {\tau}{l}p^{l+1}(1-p)^{\tau-l+1}\nonumber\\
&>&\tau\frac{1}{\sqrt {2\pi}}\frac{1}{\sqrt {\tau p(1-p)}}\frac{1}{p^{l}(1-p)^{\tau -l}}e^{-\frac{1}{6}}  p^{l+1}(1-p)^{\tau-l+1}>e^{-\frac{1}{6}}\sqrt{\frac{1-p}{2\pi}}\sqrt{l}.\nonumber
\end{eqnarray}
Thus, by setting $k=e^{-\frac{1}{6}}\sqrt{\frac{1-p}{2\pi}}$ the proof will be complete.

\section{Proof of Lemma \ref{lemorder}}\label{LemmaInv}
Lemma \ref{lemorder} states that
$C_{\text{{\sf T}}^{\text{3}}}$ can be achieved using a greedy static scheduling policy.

\begin{proof}
The proof consists of two parts. In part A we prove that when looking at class of scheduling policies that use the same assignment of packets to {\sf AP}'s for all intervals, the maximal $\text{{\sf T}}^{\text{3}}$, $R^*$, can be achieved using a greedy static scheduling policy.
In part B we prove that no policy in general can achieve any $\text{{\sf T}}^{\text{3}}$ greater than $R^*$.
Considering these two parts together, the desired result will be obtained.

\subsection{Proving that maximal $\text{{\sf T}}^{\text{3}}$ over the class of scheduling policies that use the same assignment of packets for all intervals, is achieved using a greedy static scheduling policy:}
There are a total of $N^M$ different possible ways of assigning packets to {\sf AP}'s for each interval.
We enumerate these different assignment policies by $i=1,2,\ldots, N^M$.
For an assignment $i$, $i\in\{1,2,\ldots, N^M\}$, we define $R(i)$ to be the supremum of achievable total timely throughputs, given that the assignment $i$ is used for all intervals.

Define $R^*\triangleq \max_{i\in\{1,2,\ldots, N^M\}} R(i)$.
We will now prove that there is a greedy static policy which achieves $R^*$.
It is sufficient to show that for all $i\in\{1,2,\ldots, N^M\}$ $R(i)$ can be achieved using a greedy static policy.
Consider an arbitrary $i$, $i\in\{1,2,\ldots, N^M\}$.
Since the set of packets assigned to different {\sf AP}'s are disjoint, and their channels are independent of each other, $R(i)$ is just the summation of supremums of timely throughputs on different {\sf AP}'s, when assignment $i$ is used for all intervals.

The result in \cite{Staticscheduling}  states that the timely throughput region for each {\sf AP} is a scaled version of a polymatroid (i.e., a polymatroid that has each of its coordinates scaled by a constant factor).
Moreover, in \cite{Polymatroid} it has been shown that each of the corner points of this polytope can be achieved using a static policy.
Therefore, when assignment $i$ is used, there is a static policy which achieves $R(i)$.

Furthermore, when using a static policy, according to LLN the resulting $\text{{\sf T}}^{\text{3}}$ is equal to expected number of deliveries during one interval for that static policy.
So, $R(i)$ is the highest expected number of deliveries among static  scheduling policies that use assignment policy $i$.

The following lemma implies that the highest expected number of deliveries among the static policies that use the same assignment policy is achieved by the one which serves the packets based on their channel success probabilities, and in decreasing order.
To prove this, it is sufficient to prove that for any given order if we swap the order of two adjacent clients in such a way that the client with the higher corresponding $p_{ij}$ is prioritized higher, then the expected number of deliveries will be no less than before swapping. 

\begin{lemma}\label{komaki}
Let $\tau \in \mathbb N$ and $G_1, G_2, \ldots, G_q$ be independent geometric random variables with parameters $p_1, p_2, \ldots, p_q$, respectively. Suppose that $p_d<p_{d+1}$ for some $d\in\{1,2,\ldots ,q-1\}$. In addition, let $G'_1, G'_2, \ldots, G'_q$ be independent geometric random variables (and independent of $G_i$'s) with parameters $p_1, p_2, \ldots $, $p_{d-1},p_{d+1},p_d,p_{d+2},\ldots ,p_q$, respectively.
Then, 
\begin{equation}
\sum_{i=1}^{q}\Pr(\sum_{j=1}^{i}G_j\leq \tau)\leq\sum_{i=1}^{q}\Pr(\sum_{j=1}^{i}G'_j\leq \tau).\nonumber
\end{equation}
\end{lemma}

\begin{proof}
 We have
\begin{eqnarray}
\sum_{i=1}^{q}\Pr(\sum_{j=1}^{i}G_j\leq \tau)
&=&\sum_{i=1}^{d-1}\Pr(\sum_{j=1}^{i}G_j\leq \tau)
+\Pr(\sum_{j=1}^{d}G_j\leq \tau)
+\sum_{i=d+1}^{q}\Pr(\sum_{j=1}^{i}G_j\leq \tau)\nonumber\\
&=&\sum_{i=1}^{d-1}\Pr(\sum_{j=1}^{i}G'_j\leq \tau)
+\Pr(G_d+\sum_{j=1}^{d-1}G'_j\leq \tau)
+\sum_{i=d+1}^{q}\Pr(\sum_{j=1}^{i}G'_j\leq \tau)\nonumber\\
&\stackrel{(a)}{\leq}&\sum_{i=1}^{d-1}\Pr(\sum_{j=1}^{i}G'_j\leq \tau)
+\Pr(G'_d+\sum_{j=1}^{d-1}G'_j\leq \tau)
+\sum_{i=d+1}^{q}\Pr(\sum_{j=1}^{i}G'_j\leq \tau)\nonumber\\
&=&\sum_{i=1}^{q}\Pr(\sum_{j=1}^{i}G'_j\leq \tau),\nonumber
\end{eqnarray}
where (a) follows from the fact that success probability of $G_d$, which is $p_d$, is less than success probability of $G'_d$, which is $p_{d+1}$.
\end{proof}

Lemma \ref{komaki} implies that when serving packets of some clients on an {\sf AP}, one should serve them according to their channel success probabilities, and in decreasing order in order to maximize the expected number of deliveries.
This is an intuitive fact, and Lemma \ref{komaki} formalizes this fact. 
In conclusion, $R^*$ can be achieved by a greedy static policy.

\subsection{Proving that no policy in general can achieve any $\text{{\sf T}}^{\text{3}}$ better than $R^*$:}
Consider an arbitrary scheduling policy $\eta\in \mathcal S$ (not necessarily a static policy); we will show that $\text{{\sf T}}^{\text{3}}(\eta)\leq R^*$.
Define the variable $N_j^i(k,\eta)$ to denote the outcome for client $j$ using assignment $i$ on interval $k$, i.e. if packet of client $j$ is delivered during interval $k$ when scheduling policy $\eta$ and assignment $i$ are used $N_j^i(k,\eta)=1$; otherwise $N_j^i(k,\eta)=0$.
Moreover, define function $U$ as a mapping which is used by $\eta$ from intervals to assignment policies:
\begin{equation}
U:[\mathbb N, \mathcal S]\rightarrow \{1,2,\ldots, N^M\}.\nonumber
\end{equation}
Therefore, $U(k,\eta)$ is the assignment policy used by $\eta$ for interval $k, k\in\mathbb N$.
We call $\omega=\{U(k,\eta),N_j^i(k,\eta)\}_{k=1}^{\infty}$ an outcome for policy $\eta$ over infinite intervals.
In addition, we denote the set of all possible outcomes for policy $\eta$ over infinite intervals by $\Omega(\eta)$.

In addition, define $I$ to be the set of assignments that occur infinite times. More precisely,
\begin{equation}
I\triangleq \{i\in \{1,2,\ldots, N^M\}| \forall L\in \mathbb N, \exists T\in\mathbb N \quad s.t.\quad L\leq \sum_{k=1}^{T}1(U(k,\eta)=i)\}.\nonumber
\end{equation}

According to the definition of $\text{{\sf T}}^{\text{3}}(\eta)$
\footnote[1]{$\text{{\sf T}}^{\text{3}}(\eta)=\sup \quad R\quad s.t.\quad \limsup_{T\to\infty}     \frac{\sum_{k=1}^{T}\sum_{j=1}^{M} \sum_{i=1}^{N^M}N_j^i(k,\eta)}{T}\geq R$ with probability $1$.}
there exists a subset of $\Omega(\eta)$, denoted by $A$, such that $P(A)=1$ and for all  
$\omega=\{U(k,\eta),N_j^i(k,\eta)\}_{k=1}^{\infty}$ and $\omega\in A$, 
$$\text{{\sf T}}^{\text{3}}(\eta)\leq  \limsup_{T\to\infty} (    \frac{\sum_{k=1}^{T}\sum_{j=1}^{M} \sum_{i=1}^{N^M}N_j^i(k,\eta)}{T}).$$

Therefore, for any outcome $\omega=\{U(k,\eta),N_j^i(k,\eta)\}_{k=1}^{\infty}\in A$, we have 
\begin{align}
\text{{\sf T}}^{\text{3}}(\eta)&\leq  \limsup_{T\to\infty} (     \frac{\sum_{k=1}^{T}\sum_{j=1}^{M}\sum_{i=1}^{N^M}N_j^i(k,\eta)}{T}) \stackrel{(a)}{=}      \limsup_{T\to\infty}(\frac{\sum_{k=1}^{T}\sum_{j=1}^{M}\sum_{i\in I}^{}N_j^i(k,\eta)}{T})\nonumber\\
&\stackrel{(b)}{=}      \limsup_{T\to\infty}(\sum_{i\in I}^{}(\frac{{\sum_{k=1}^{T}1(U(k,\eta)=i)}}{T})\times(\frac{\sum_{k=1}^{T}\sum_{j=1}^{M}N_j^i(k,\eta)}{\sum_{k=1}^{T}1(U(k,\eta)=i)})),\label{lastline}
\end{align}
where (a) follows from the fact that the assignment $i$, where $i\notin I$, does not contribute to the value of limsup according to the definition of $I$.
In addition, (b) is true because the fraction $\frac{\sum_{k=1}^{T}\sum_{j=1}^{M}\sum_{i\in I}^{}N_j^i(k,\eta)}{\sum_{k=1}^{T}1(U(k,\eta)=i)}$ is properly defined for $i\in I$ since its denominator is not zero as $T\to\infty$.
The reason why the denominator is not zero as $T\to\infty$ is that there exists $r\in \mathbb N$ such that $\sum_{k=1}^{r}1(U(k,\eta)=i) \geq 1$ for $i\in I$ according to the definition of $I$. This means that for $T>r$, the fraction is well-defined.

Moreover, since $\limsup_{T\to\infty}(\frac{\sum_{k=1}^{T}\sum_{j=1}^{M}\sum_{i\in I}^{}N_j^i(k,\eta)}{\sum_{k=1}^{T}1(U(k,\eta)=i)})$ is the average number of successful deliveries for intervals for which assignment $i$ is applied, 
there exists a subset of $\Omega(\eta)$, denoted by $B$, such that $P(B)>0$ and for all  
$\omega=\{U(k,\eta),N_j^i(k,\eta)\}_{k=1}^{\infty}$, $\omega\in B$,
 $$\limsup_{T\to\infty}(\frac{\sum_{k=1}^{T}\sum_{j=1}^{M}N_j^i(k,\eta)}{\sum_{k=1}^{T}1(U(k,\eta)=i)})\leq R^*.
\footnote[1]{This is true because if $\limsup_{T\to\infty}(\frac{\sum_{k=1}^{T}\sum_{j=1}^{M}N_j^i(k,\eta)}{\sum_{k=1}^{T}1(U(k,\eta)=i)})> R^*$ with probability 1, then we have a scheduling policy which uses the same assignment policy for all intervals, and achieves a $\text{{\sf T}}^{\text{3}}$ which is strictly greater than $R^*$ which is not possible according to the result in part A of the proof.}$$

In addition, note that $P(A\cap B)=P(A)-P(A\cup B)+P(B)=P(B)>0$, which means $A\cap B$ is not empty.
Hence, using (\ref{lastline}) there is an outcome of $\eta$,  
$\omega=\{U(k,\eta),N_j^i(k,\eta)\}_{k=1}^{\infty}$ and $\omega\in A\cap B$, for which 
\begin{align}
\text{{\sf T}}^{\text{3}}(\eta)&\leq  \limsup_{T\to\infty}(\sum_{i\in I}^{}(\frac{{\sum_{k=1}^{T}1(U(k,\eta)=i)}}{T})\times(\frac{\sum_{k=1}^{T}\sum_{j=1}^{M}N_j^i(k,\eta)}{\sum_{k=1}^{T}1(U(k,\eta)=i)}))\nonumber\\
&\stackrel{(c)}{\leq} \limsup_{T\to\infty}(\sum_{i\in I}^{}\frac{\sum_{k=1}^{T}1(U(k,\eta)=i)}{T})  \times R^* \stackrel{(d)}{\leq}R^*,\nonumber
\end{align}
where (c) follows from the fact that for $\omega=\{U(k,\eta),N_j^i(k,\eta)\}_{k=1}^{\infty}$, and $\omega\in A\cap B$,\\ 
$\limsup_{T\to\infty}(\frac{\sum_{k=1}^{T}\sum_{j=1}^{M}N_j^i(k,\eta)}{\sum_{k=1}^{T}1(U(k,\eta)=i)})\leq R^*$, and also using Lemma \ref{limsup}.
Finally (d) follows from the fact that for each interval the scheduling policy can choose at most one of the $N^M$ different possible assignments, or in other words, $\sum_{i\in I}^{}\frac{\sum_{k=1}^{T}1(U(k,\eta)=i)}{T}\leq 1$ for all $T\in\mathbb N$.

Therefore, for scheduling policy $\eta$, $\text{{\sf T}}^{\text{3}}(\eta)\leq R^*$.
Using Part A and Part B we conclude that $C_{\text{{\sf T}}^{\text{3}}}=R^*$, and $C_{\text{{\sf T}}^{\text{3}}}$ can be achieved using a greedy static policy. 
\end{proof}

Below we provide Lemma \ref{limsup} and its proof.

\begin{lemma}\label{limsup}
Suppose $L$ is an integer, and $\{A_{1T}\}_{T=1}^{\infty},\{A_{2T}\}_{T=1}^{\infty},\ldots, \{A_{LT}\}_{T=1}^{\infty}$, and\\ $\{B_{1T}\}_{T=1}^{\infty},\{B_{2T}\}_{T=1}^{\infty},\ldots, \{B_{LT}\}_{T=1}^{\infty}$ are non-negative real sequences, where\\
$\limsup_{T\to\infty}\sum_{i=1}^{L}A_{iT}<\infty,$
and 
for any $i\in\{1,2,\ldots ,L\},$  $\limsup_{T\to\infty}B_{iT}\leq B.$
Then, 
\begin{equation}
\limsup_{T\to\infty}\sum_{i=1}^{L}A_{iT}B_{iT}\leq ( \limsup_{T\to\infty}\sum_{i=1}^{L}A_{iT})\times B. \nonumber
\end{equation}
\end{lemma}

\begin{proof}
Consider an arbitrary $\epsilon>0$.
Since $\forall i\in\{1,2,\ldots, L\} \quad \limsup_{T\to\infty}B_{iT}\leq B$, 
$$\exists M\in \mathbb N,\quad s.t.\quad \forall i\in\{1,2,\ldots, L\},T\geq M\quad B_{iT}\leq B+\epsilon.$$

Therefore, for all $r\geq M$ we will have
$\sup_{T\geq r}\sum_{i=1}^{L}A_{iT}B_{iT}\leq \sup_{T\geq r}\sum_{i=1}^{L}A_{iT}(B+\epsilon).$
Hence,
$\lim_{r\to\infty}\sup_{T\geq r}\sum_{i=1}^{L}A_{iT}B_{iT}\leq (B+\epsilon)\lim_{r\to\infty}\sup_{T\geq r}\sum_{i=1}^{L}A_{iT}.$
Since the inequality is true for any $\epsilon >0$, we have
$\limsup_{T\to\infty}\sum_{i=1}^{L}A_{iT}B_{iT}\leq (\limsup_{T\to\infty}\sum_{i=1}^{L}A_{iT})\times B.$
\end{proof}

\section{Proof of Lemma \ref{lemmaup}}\label{rightlemma}

For $l=0$, we have $p_i<\frac{1}{\tau}$, for $i=1,2,\ldots, q$.
Therefore, $E[Y]$ in this case is less than that of the case in which $p_1=p_2=\ldots =p_{q}=\frac{1}{\tau}$. 
On the other hand, for $p_1=p_2=\ldots =p_{q}=\frac{1}{\tau}$  $E[Y]\leq\tau \times\frac{1}{\tau}=1.$
Hence, the statement is true for $l=0$.
Now, suppose that $l>0$.
We know that $l= \max \quad \hat l\quad s.t.\quad \sum_{i=1}^{\hat l}1/p_i\leq \tau$.
Therefore, we have 
\begin{equation}
\sum_{i=1}^{l}\frac{1}{p_i}\leq \tau< \sum_{i=1}^{l+1}\frac{1}{p_i}.\nonumber
\end{equation}
We will show that $E[Y]$ can be at most $l+1$.
Without loss of generality we can omit $p_i$'s that are equal to zero; because by omitting them neither of $E[Y]$ nor $l$ change, and $E[Y]-l$ would remain the same. So, we suppose that $1\geq p_1\geq p_2\geq \ldots\geq p_{q}> 0$. 
It is sufficient to prove the lemma for the case of $q=\tau$; because if we have less than $\tau$ geometric random variables, $E[Y]$ will be less. On the other hand, we do not need to consider the case $q>\tau$; since for $i>\tau$, $\Pr(\sum_{j=1}^{i}G_j\leq \tau)=0$.
 Therefore, we suppose that $q=\tau$.

Let $X_l=\sum_{i=1}^{l} G_i$, where $G_i=Geom(p_i)$.
By this notation we have:
\begin{equation}
\quad\Pr(X_{l}>\tau)=\sum_{i=0}^{l-1} \Pr(Y=i).\nonumber
\end{equation}

Now we write down the expression for $E[Y]$:
\begin{equation}\label{expectt1}
E[Y]=\sum_{i=0}^{l-1} i\Pr(Y=i)+\sum_{t=l}^{\tau}\Pr(X_{l}=t)(\sum_{i=l}^{\tau}i\Pr(Y=i|X_{l}=t)).\footnote[1]{Note that the summation on $t$ here is from $l$ to $\tau$;  since for $0\leq t<\ell$, $\Pr(X_{l}=t)=0$.}
\end{equation}
Since $1\geq p_l\geq p_{l+1}\geq \ldots\geq p_{\tau}>0$, $E[Y]$ is less than the case where $p_l=p_{l+1}=\ldots=p_{\tau}$, although $l$ remains the same. So it is sufficient to prove Theorem 1 for the case where $p_l=p_{l+1}=\ldots=p_{\tau}$. 
For $t\leq \tau$ if we set $p_l=p_{l+1}=\ldots=p_{\tau}$ we have
\begin{equation}
\sum_{i=l}^{\tau}i\Pr(Y=i|X_{l}=t)=E[Y|X_{l}=t]=l+(\tau-t)p_l.\label{expectt2}
\end{equation}
Therefore, by using ($\ref{expectt1}$) and ($\ref{expectt2}$) we have 
\begin{align}
&E[Y]= \sum_{i=0}^{l-1} i\Pr(Y=i) + \sum_{t=l}^{\tau} (l+p_l(\tau-t))\Pr(X_{l}=t)\nonumber\\
&= \sum_{i=0}^{l-1} i\Pr(Y=i)+( l+p_l\tau)(1-\Pr(X_{l}>\tau))-  p_l[\sum_{t=l}^{\infty}t\Pr(X_l=t)-\sum_{t=\tau+1}^{\infty}t\Pr(X_l=t)]\nonumber
\end{align}
\begin{align}
&=\sum_{i=0}^{l-1} i\Pr(Y=i)+ (l+p_l\tau)-(l+p_l\tau)\sum_{i=0}^{l-1}\Pr(Y=i)-p_l\sum_{i=1}^{l}\frac{1}{p_i}+p_l\sum_{t=\tau+1}^{\infty} t\Pr(X_{l}=t)\nonumber\\
&=\sum_{i=0}^{l-1}( i-l-p_l\tau)\Pr(Y=i)+(l+p_l(\tau-\sum_{i=1}^{l}\frac{1}{p_i}))+p_l\sum_{t=\tau+1}^{\infty} t\Pr(X_{l}=t)\nonumber\\
&\stackrel {(a)}{<}\sum_{i=0}^{l-1}( i-l-p_l\tau)\Pr(Y=i)+l+1+p_l\sum_{t=\tau+1}^{\infty} t\Pr(X_{l}=t),\label{asb}
\end{align}
where the last inequality (a) follows from $\tau<\sum_{i=1}^{l+1}\frac{1}{p_i}$ and the assumption that $p_{l+1}= p_l$.
Now, we only need to rewrite $p_l\sum_{t=\tau+1}^{\infty} t\Pr(X_{l}=t)$ in terms of $Y$.
For $t>\tau$ we have
\begin{equation}
\Pr(X_{l}=t)=\sum_{i=0}^{l-1}\Pr(X_{l}=t|Y=i)\Pr(Y=i).\nonumber
\end{equation}
Therefore, 
\begin{eqnarray}
\sum_{t=\tau+1}^{\infty} t\Pr(X_{l}=t)&=&\sum_{t=\tau+1}^{\infty}t(\sum_{i=0}^{l-1}\Pr(X_{l}=t|Y=i)\Pr(Y=i))\nonumber\\
&=&\sum_{i=0}^{l-1} \Pr(Y=i)(\sum_{t=\tau+1}^{\infty} t\Pr(X_{l}=t|Y=i)).\nonumber
\end{eqnarray}
But due to memoryless property of geometric distribution, we know that 
\begin{align}
&\sum_{t=\tau+1}^{\infty} (t-\tau)\Pr(X_{l}=t|Y=i)=\sum_{t=\tau+1}^{\infty} (t-\tau)\Pr(\sum_{j=i+1}^{l} G_j=t-\tau)\nonumber\\
&=\sum_{t=1}^{\infty} t\Pr(\sum_{j=i+1}^{l} G_j=t)=\sum_{j=i+1}^{l} \frac{1}{p_j},\qquad\forall i\leq l-1. \nonumber
\end{align}
Therefore,
$\sum_{t=\tau+1}^{\infty} t\Pr(X_{l}=t|Y=i)=\tau+\sum_{j=i+1}^{l} \frac{1}{p_j}.$
Hence,
\begin{equation}\label{equality}
p_l\sum_{t=\tau+1}^{\infty} t\Pr(X_{l}=t)=\sum_{i=0}^{l-1} \Pr(Y=i)(p_l\tau+\sum_{j=i+1}^{l} \frac{p_l}{p_j}).
\end{equation}
Substituting (\ref{equality}) into (\ref{asb}) we get
\begin{eqnarray}
E[Y]&< &\sum_{i=0}^{l-1}( i-l-p_l\tau)\Pr(Y=i)+l+1+\sum_{i=0}^{l-1} \Pr(Y=i)(p_l\tau+\sum_{j=i+1}^{l} \frac{p_l}{p_j})\nonumber\\
&=&l+1+\sum_{i=0}^{l-1} \Pr(Y=i)(i-l-p_l\tau+p_l\tau+\sum_{j=i+1}^{l} \frac{p_l}{p_j})\leq l+1,\nonumber
\end{eqnarray}
where the last inequality follows from the fact that $\forall j\in\{i+1,\ldots ,l\}\quad p_l\leq p_j.$

\section{Proof of Lemma \ref{lemmadown}}\label{leftlemma}

We will show that $E[Y]>l-2\sqrt{l+\frac{1}{4}}$.
It is sufficient to prove Lemma \ref{lemmadown} for  $q=l$; because for $q>l$, $E[Y]$ would only increase. On the other hand, $q$ cannot be less than $l$ according to the assumption $l= \max \hat l\quad s.t.\quad \sum_{i=1}^{\hat l}1/p_i\leq \tau$. Therefore, from now on we suppose $q=l$. By our notation we have
\begin{equation}\label{GY}
\quad\Pr(\sum_{j=1}^{i}G_j>\tau)=\sum_{j=0}^{i-1} \Pr(Y=j),\quad i=1,2,\ldots ,l.
\end{equation}
We now bound $l-E[Y]$ from above.
\begin{eqnarray}
l-E[Y]&=&l-\sum_{i=1}^{l}\Pr(Y\geq i)=\sum_{i=1}^{l}(1- \Pr(Y\geq i))=\sum_{i=1}^{l}\Pr(Y<i)\nonumber\\
&\stackrel{(a)}{=}&\sum_{i=1}^{l}\Pr(\sum_{j=1}^{i}G_j>\tau)\stackrel{(b)}{\leq}\sum_{i=1}^{l}\Pr(\sum_{j=1}^{i} G_j>\sum_{j=1}^{l}\frac{1}{p_j})\nonumber\\
&\leq&1+\sum_{i=1}^{l-1}\Pr(|\sum_{j=1}^{i} (G_j-\frac{1}{p_j})|>\sum_{j=i+1}^{l}\frac{1}{p_j})\nonumber \\
&\stackrel{(c)}{\leq}&1+\sum_{i=1}^{l-1}    \min(1,\frac{\mathrm{var}(\sum_{j=1}^{i} G_j)}{(\sum_{j=i+1}^{l}\frac{1}{p_j})^2})      \stackrel{(d)}{\leq } 1+\sum_{i=1}^{l-1}\min(1,\frac{\sum_{j=1}^{i}\frac{1}{p_j^2}}{(\sum_{j=i+1}^{l}\frac{1}{p_j})^2}),\nonumber
\end{eqnarray}
where (a) follows from (\ref{GY});  (b) follows from $\sum_{i=1}^{ l}1/p_i\leq \tau$; 
 (c) follows from Chebyshev's inequality, where $\mathrm{var}(\sum_{j=1}^{i} G_j)$ is the variance of the random variable $\sum_{j=1}^{i} G_j$;
and (d) follows due to independence of $G_i$'s, which results in
$\mathrm{var}(\sum_{j=1}^{i} G_j)=\sum_{j=1}^{i}\mathrm{var}( G_j)=\sum_{j=1}^{i} \frac{1-p_j}{p_j^2}< \sum_{j=1}^{i} \frac{1}{p_j^2},\quad  i=1,2,\ldots,l.$
But since $p_1\geq p_2\geq \ldots\geq p_l$, we have $\sum_{j=1}^{i} \frac{1}{p_j^2}\leq \frac{i}{p_i^2}$ and $(\sum_{j=i+1}^{l} \frac{1}{p_j})^2\geq (\frac{l-i}{p_i})^2$.
Therefore, 
\begin{equation}\label{leymin}
l-E[Y]\leq 1+\sum_{i=1}^{l-1} \min(1,\frac{\frac{i}{p_i^2}}{\frac{(l-i)^2}{p_i^2}})=1+\sum_{i=1}^{l-1} \min(1,\frac{i}{(l-i)^2}).
\end{equation}
Hence, by (\ref{leymin}) and applying Lemma \ref{ell}
the proof of Lemma \ref{lemmadown} will be complete.
\begin{lemma}\label{ell} Assume $l\in \mathbb{N}$, and $l>1$. Then, 
$1+\sum_{i=1}^{l-1} \min(1,\frac{i}{(l-i)^2})<2\sqrt {l+\frac{1}{4}}\label{equell}.$
\end{lemma}

\begin{proof} 
For $l<18$ the statement of the Lemma can be verified numerically. Therefore, suppose that $l\geq 18$.
Let $f(i)\triangleq \frac{i}{(l-i)^2}$, for $i \in \mathbb {N}, 1\leq i\leq l-1$; and consider the following three observations regarding the function $f(.)$:
\begin{enumerate}
\item $f(i)$ increases as $i$ increases for $i \in \mathbb {N}, 1\leq i\leq l-1. $
\item $f(1)=\frac{1}{(l-1)^2}<1$.
\item $f(l-1)=\frac{l-1}{1}>1$.
\end{enumerate}
Therefore, $ \exists m\in \mathbb {N},1\leq m< l-1 \text{ such that }$
\begin{equation}
\begin{split}
\frac{m}{(l-m)^2}\leq 1<\frac{m+1}{(l-(m+1))^2}\label{z}.
\end{split}
\end{equation}
Note that $m\ne l-1$, because $\frac{l-1}{(l-(l-1))^2}>1$.
We  rewrite the inequalities in (\ref{z}) as
\begin{equation}
l-\sqrt {l+\frac{1}{4}}-\frac{1}{2}<m\leq l-\sqrt {l+\frac{1}{4}}+\frac{1}{2}\label{ez}.
\end{equation}
In addition,
\begin{equation}
1+\sum_{i=1}^{l-1} \min(1,\frac{i}{(l-i)^2})=1+\sum_{i=1}^{m} \min(1,\frac{i}{(l-i)^2})+\sum_{i=m+1}^{l-1} \min(1,\frac{i}{(l-i)^2}).\nonumber
\end{equation}
But from (\ref{z}) and the fact that $f(i)=\frac{i}{(l-i)^2}$ increases by increase of $i$, we have
\begin{align}
1&+\sum_{i=1}^{l-1} \min(1,\frac{i}{(l-i)^2})=1+\sum_{i=1}^{m} \frac{i}{(l-i)^2}+(l-1-m)=l-m+\sum_{j=l-m}^{l-1} \frac{l-j}{j^2}\nonumber\\
&< l-m+\sum_{j=l-m}^{l-1} \frac{l-j}{j(j-1)}
= l-m+\sum_{j=l-m}^{l-1} (\frac{l-j}{j-1}-\frac{l-j}{j})= l-m+\frac{m}{l-m-1}-\sum_{j=l-m}^{l-1}\frac{1}{j}\nonumber\\
&\stackrel{(a)}{<} l-m+\frac{m}{l-m-1}-\frac{m}{l-\frac{m+1}{2}}=l-m+\frac{m(m+1)}{(l-m-1)(l-m+l-1)}\nonumber\\
&\stackrel{(b)}{<} (\sqrt{l+\frac{1}{4}}+\frac{1}{2})+\frac{(l-\sqrt{l+\frac{1}{4}}+\frac{1}{2})(l-\sqrt{l+\frac{1}{4}}+\frac{3}{2})}
{(\sqrt{l+\frac{1}{4}}-\frac{3}{2})(l+\sqrt{l+\frac{1}{4}}-\frac{3}{2})}\nonumber\\
&=(\sqrt{l+\frac{1}{4}}+\frac{1}{2})+(\sqrt{l+\frac{1}{4}}-\frac{3}{2}+
\frac{5l-9\sqrt{l+\frac{1}{4}}+\frac{11}{2}}
{(\sqrt{l+\frac{1}{4}}-\frac{3}{2})(l+\sqrt{l+\frac{1}{4}}-\frac{3}{2})})\nonumber\\
&= 2\sqrt{l+\frac{1}{4}}+\frac{-l\sqrt{l+\frac{1}{4}}+\frac{11}{2}l-6\sqrt{l+\frac{1}{4}}+3}
{(\sqrt{l+\frac{1}{4}}-\frac{3}{2})(l+\sqrt{l+\frac{1}{4}}-\frac{3}{2})}\label{endell},
\end{align}
where (a) follows from the Cauchy-Schwarz inequality; and (b) follows from (\ref{ez}).
For $l\geq 18$ the term $\frac{-l\sqrt{l+\frac{1}{4}}+\frac{11}{2}l-6\sqrt{l+\frac{1}{4}}+3}
{(\sqrt{l+\frac{1}{4}}-\frac{3}{2})(l+\sqrt{l+\frac{1}{4}}-\frac{3}{2})}$ in (\ref{endell}) is less than zero.
Therefore, the statement of Lemma \ref{ell} is true for all $l>1,l\in\mathbb N$.
\end{proof}

\section{Proof of Corollary \ref{run}}\label{corollary}

Let $ \vec\Pi^*$ denote the partition (assignment) chosen by the optimal greedy static scheduling policy $\eta_{\text{g-static}}^*$. 
Therefore, we have $||\vec R( \eta_{\text{g-static}}^*)||_1=C_{\text{{\sf T}}^{\text{3}}}$. 
Furthermore, consider an assignment, denoted by $\vec\Pi_{det}$, which maximizes the objective function in (\ref{relaxed}). 
Let $\eta_{\text{g-static}}^{\text{det}}$ denote the greedy static scheduling policy which corresponds to $\vec\Pi_{det}$.
Further, let $||\vec R_{\text{det}}( \eta_{\text{static}})||_1$ designate the maximum number of objects that can be packed in the RP in (\ref{relaxed}) when a static scheduling policy $ \eta_{\text{static}}$ is implemented. 
Therefore, $||\vec R_{\text{det}}( \eta_{\text{g-static}}^{\text{det}})||_1=C_{\text{det}}$, since 
 $||\vec R_{\text{det}}( \eta_{\text{g-static}}^{\text{det}})||_1$ is the value of the objective function in (\ref{relaxed}) when the assignment is dictated by $\eta_{\text{g-static}}^{\text{det}}$. 
 The right part of the inequality in Corollary \ref{run} in (\ref{yilu}) is trivial since $C_{\text{{\sf T}}^{\text{3}}}$ is the optimal $\text{\sf T}^{\text{3}}$ achievable under any scheduling policy. So we only need to prove the left part of the inequality in (\ref{yilu}).
Using a similar argument as the one in part B of Section \ref{theorem1}, and by applying Cauchy-Schwarz inequality, we get
\begin{equation}\label{aaa}
||\vec R(  \eta_{\text{g-static}}^{\text{det}})||_1\geq ||\vec R_{\text{det}}( \eta_{\text{g-static}}^{\text{det}})||_1-2\sqrt {N(||\vec R_{\text{det}}( \eta_{\text{g-static}}^{\text{det}})||_1+\frac{N}{4})}.
\end{equation}
Now consider the function $g(.)$ defined as follows: $g(x)\triangleq x-2\sqrt{(N(x+\frac{N}{4}))},\quad x\in\mathbb R$.
\\So, $g(x)$ is strictly increasing for $x>\frac{3N}{4}$. On the other hand, we know that 
\begin{equation}\label{aab}
C_{\text{det}}=||\vec R_{\text{det}}( \eta_{\text{g-static}}^{\text{det}})||_1 \geq ||\vec R_{\text{det}}(  \eta_{\text{g-static}}^*)||_1\geq ||\vec R( \eta_{\text{g-static}}^*)||_1-N=C_{\text{{\sf T}}^{\text{3}}}-N,
\end{equation}
where  the right inequality follows from Theorem \ref{maintheorem}. By $g(x)$ being an increasing function of $x$ and (\ref{aab}) we get
\begin{equation}\label{aac}
||\vec R_{\text{det}}( \eta_{\text{g-static}}^{\text{det}})||_1-2\sqrt {N(||\vec R_{\text{det}}(  \eta_{\text{g-static}}^{\text{det}})||_1+\frac{N}{4})}\geq C_{\text{{\sf T}}^{\text{3}}}-N-2\sqrt {N(C_{\text{{\sf T}}^{\text{3}}}-\frac{3N}{4})}.
\end{equation}
Hence, by (\ref{aaa}) and (\ref{aac}) we get
$||\vec R(  \eta_{\text{g-static}}^{\text{det}})||_1\geq C_{\text{{\sf T}}^{\text{3}}}-N-2\sqrt {N(C_{\text{{\sf T}}^{\text{3}}}-\frac{3N}{4})}. $

\section{Proof of Theorem \ref{weightedtheorem}}\label{weightedproof}
By the same argument as in proof of Lemma \ref{lemorder}, $C_{w\text{-{\sf T}}^{\text{3}}}$ can be achieved by a static scheduling policy. 
Therefore, by LLN, to achieve $C_{w\text{-{\sf T}}^{\text{3}}}$, it is sufficient to find the assignment and ordering which provide the highest expected weighted delivery for one interval.
First, we show that for a given assignment $\vec\Pi=[\mathcal I_1,\mathcal I_2,\ldots ,\mathcal I_N]$ the optimal ordering of the packets of clients assigned to ${\sf AP}_i$ is according to the order of $\omega_j p_{ij}$, $j\in\mathcal I_i$. 
To do so, it is sufficient to prove that for any given order of the clients if we swap two adjacent clients such that the client with higher $\omega_jp_j$ is prioritized higher, then the expected weighted delivery will be no less than before swapping. 
The following lemma formally states this fact. 
\begin{lemma}\label{wkomaki}
Let $\tau,q \in \mathbb N$, and $\omega_1,\omega_2,\ldots ,\omega_q\in \mathbb R$.
Also, for some $d\in\{1,2,\ldots ,q-1\}$, let $\omega'_i=\omega_i$, for $1\leq i<d$ and $d+1<i\leq q$;  and $\omega'_d=\omega_{d+1}$, $\omega'_{d+1}=\omega_d$.
Further, let $G_1, G_2, \ldots, G_q$ be independent geometric random variables with parameters $p_1, p_2, \ldots, p_q$, respectively. Suppose that $\omega_d p_d<\omega_{d+1}p_{d+1}$. In addition, let $G'_1, G'_2, \ldots, G'_q$ be independent geometric random variables, independent of $G_i$'s, with parameters $p_1, p_2, \ldots $, $p_{d-1},p_{d+1},p_d,p_{d+2},\ldots ,p_q$, respectively.
Then, 
\begin{equation}
\sum_{i=1}^{q}\omega_i\Pr(\sum_{j=1}^{i}G_j\leq \tau)\leq\sum_{i=1}^{q}\omega'_i\Pr(\sum_{j=1}^{i}G'_j\leq \tau).\nonumber
\end{equation}
\end{lemma}

\begin{proof}
Let $A=\sum_{i=1}^{q}\omega_i\Pr(\sum_{j=1}^{i}G_j\leq \tau)$, and $ B=\sum_{i=1}^{q}\omega'_i\Pr(\sum_{j=1}^{i}G'_j\leq \tau)$. Then,
\begin{align}
B-A&=\sum_{i=1}^{d-1}\omega'_i\Pr(\sum_{j=1}^{i}G'_j\leq \tau)+
\omega'_d\Pr(\sum_{j=1}^{d}G'_j\leq \tau)+\omega'_{d+1}\Pr(\sum_{j=1}^{d+1}G'_j\leq \tau)\nonumber\\
&+\sum_{i=d+2}^{q}\omega'_i\Pr(\sum_{j=1}^{i}G'_j\leq \tau)
-\sum_{i=1}^{d-1}\omega_i\Pr(\sum_{j=1}^{i}G_j\leq \tau)-
\omega_d\Pr(\sum_{j=1}^{d}G_j\leq \tau)\nonumber\\
&-\omega_{d+1}\Pr(\sum_{j=1}^{d+1}G_j\leq \tau)-
\sum_{i=d+2}^{q}\omega_i\Pr(\sum_{j=1}^{i}G_j\leq \tau)\nonumber\\
&=\omega'_d\Pr(\sum_{j=1}^{d}G'_j\leq \tau)+\omega'_{d+1}\Pr(\sum_{j=1}^{d+1}G'_j\leq \tau)-\omega_d\Pr(\sum_{j=1}^{d}G_j\leq \tau)\nonumber\\
&-\omega_{d+1}\Pr(\sum_{j=1}^{d+1}G_j\leq \tau)\nonumber\\
&=\sum_{t=1}^{\tau}\Pr(\sum_{j=1}^{d-1}G'_j=t)[\omega'_d\Pr(G'_d\leq \tau-t)+\omega'_{d+1}\Pr(G'_{d}+G'_{d+1}\leq \tau-t)]\nonumber\\
&-\sum_{t=1}^{\tau}\Pr(\sum_{j=1}^{d-1}G_j=t)[\omega_d\Pr(G_d\leq \tau-t)+\omega_{d+1}\Pr(G_{d}+G_{d+1}\leq \tau-t)]\nonumber
\end{align}
\begin{align}
&=\sum_{t=1}^{\tau}\Pr(\sum_{j=1}^{d-1}G_j=t)[\omega'_d\Pr(G'_d\leq \tau-t)+\omega'_{d+1}\Pr(G'_{d}+G'_{d+1}\leq \tau-t)\nonumber\\
&-\omega_d\Pr(G_d\leq \tau-t)-\omega_{d+1}\Pr(G_{d}+G_{d+1}\leq \tau-t)]\nonumber.
\end{align}
Therefore, it is sufficient to show that for all $t\in\mathbb N$,
\begin{equation}
\omega'_d\Pr(G'_d\leq t)+\omega'_{d+1}\Pr(G'_{d}+G'_{d+1}\leq t)-\omega_d\Pr(G_d\leq t)-\omega_{d+1}\Pr(G_{d}+G_{d+1}\leq t)\geq 0.\nonumber
\end{equation}
Note that 
\begin{itemize}
\item $\omega'_d=\omega_{d+1}$, and $\omega'_{d+1}=\omega_d$.
\item $\Pr(G'_d\leq t)=1-(1-p_{d+1})^t$, and $\Pr(G_d\leq t)=1-(1-p_{d})^t$.
\item $\Pr(G_{d}+G_{d+1}\leq t)=\Pr(G'_{d}+G'_{d+1}\leq t)=1-\frac{p_d(1-p_{d+1})^t-p_{d+1}(1-p_d)^t}{p_d-p_{d+1}}$.
\end{itemize}

Therefore,
\begin{align}
&\omega'_d\Pr(G'_d\leq t)+\omega'_{d+1}\Pr(G'_{d}+G'_{d+1}\leq t)-\omega_d\Pr(G_d\leq t)-\omega_{d+1}\Pr(G_{d}+G_{d+1}\leq t)\\
&=(\omega_{d+1}p_{d+1}-\omega_dp_d)(\frac{(1-p_{d+1})^{t}-(1-p_d)^{t}}{p_d-p_{d+1}})>0,\quad t\in \mathbb N,
\end{align}
where the inequality follows from the assumption that $\omega_{d+1}p_{d+1}-\omega_dp_d>0$.
\end{proof}

\subsection{Proof of $C_{w\text{-{\sf T}}^{\text{3}}}<C_{w\text{-det}}+N\omega_{max}$}
We follow the same line of proof as in Section \ref{theorem1}.
Since  $C_{w\text{-\sf T}^{\text{3}}}$ can be achieved using a static scheduling policy which uses ordering according to $\omega_jp_j$'s, it is sufficient to show that for any static scheduling policy $\eta_{\text{wg-static}}$ which uses its corresponding optimal ordering we have
$w\text{-\sf T}^{\text{3}}(\eta_{\text{wg-static}})<C_{w\text{-det}}+N\omega_{max}.$
Suppose an arbitrary static scheduling policy $\eta_{\text{wg-static}}$ with the corresponding partition $\vec\Pi_{\text{wg-static}}=[\mathcal{I}_1,\mathcal{I}_2,\ldots ,\mathcal{I}_N]$, which uses the optimal ordering  is implemented. 
By (\ref{wt3rj}) we know that 
$w\text{-\sf T}^{\text{3}}(\eta_{\text{wg-static}})=\sum_{j=1}^{M}\omega_jR_j(\eta_{\text{wg-static}}).$
On the other hand for $j\in[1:M]$, by (\ref{liminf}) we have
$R_j(\eta_{\text{wg-static}})=\limsup_{r\to\infty}\frac{\sum_{k=1}^{r}N_j(k,\eta_{\text{wg-static}})}{r}.$
For $i\in [1:N]$ define 
$Y_i\triangleq \sum_{j\in \mathcal I_i}^{}N_j(1,\eta_{\text{wg-static}})$
and $q_i\triangleq|\mathcal I_i|$. 
Denote the enumeration of clients assigned to $\text{{\sf AP}}_i$ by
$\{\mathcal I_i(1),\mathcal I_i(2),\ldots,\mathcal I_i(q_i)\}$,
where the enumeration is according to the optimal ordering for the weighted case.
Since a static scheduling policy is implemented and channels are i.i.d over time, by LLN we have
\begin{equation}
R_{\mathcal I_i(j)}(\eta_{\text{wg-static}})=\limsup_{r\to\infty}\frac{\sum_{k=1}^{r}N_{\mathcal I_i(j)}(k,\eta_{\text{wg-static}})}{r}
=\Pr(Y_i\geq j),\quad 1\leq j\leq q_i,\quad 1\leq i\leq N.\nonumber
\end{equation} 
Therefore, it is easy to see that $w\text{-\sf T}^{\text{3}}(\eta_{\text{wg-static}}) =\sum_{i=1}^{N}\sum_{j=1}^{q_i}(\sum_{k=1}^{j}\omega_{\mathcal I_i(k)})\Pr(Y_i=j)$.
Let $G_{ij}$ be a geometric random variable with parameter $p_{ij}, i\in[1:N], j\in[1:M]$.
Then, for $i\in[1:N]$, $1\leq k\leq q_i$,
$Y_i= \max\quad k\quad s.t. \quad \sum_{j=1}^{k} G_{i\mathcal I_i(j)}\leq \tau, $
since $\eta_{\text{wg-static}}$ persistently sends a packet until it is delivered, or the interval is over.
The following lemma, which is the generalized version of Lemma \ref{lemmaup}, relates $l_i$ and $\omega_j$'s to $Y_i$. 

\begin{lemma}\label{weightedlemmaup}
Let $1\leq\omega_1,\omega_2,\ldots ,\omega_q\leq \omega_{max}$ for some $\omega_{max}\in\mathbb R$.
Also, let $\tau \in \mathbb N$ and $G_1, G_2, \ldots, G_q$ be independent geometric random variables with parameters $p_1, p_2, \ldots, p_q$ respectively, such that $\omega_1p_1 \geq \omega_2 p_2 \geq \ldots \geq \omega_qp_q \geq 0$. Also define
$l\triangleq \max \hat l\quad s.t.\quad \sum_{i=1}^{\hat l}1/p_i\leq \tau,$
 and
$Y\triangleq \max i\quad s.t. \quad \sum_{j=1}^{i} G_j\leq \tau, \quad i\in\{1,2,\ldots ,q\}.$
Then, we have
$\sum_{i=1}^{q}(\sum_{j=1}^{i}\omega_j)\Pr(Y=i)<\sum_{j=1}^{l}\omega_j+\omega_{max}.$
\end{lemma}

\begin{proof}
 Suppose that $l>0$ (for $l=0$ the proof is straightforward).
We have 
\begin{equation}
\sum_{i=1}^{l}\frac{1}{p_i}\leq \tau< \sum_{i=1}^{l+1}\frac{1}{p_i}.
\end{equation}
 Without loss of generality we can omit $p_i$'s that are equal to zero and assume $0 <p_1, p_2, \ldots, p_q\leq 1$. 
Furthermore, according to the same argument as in proof of Theorem 1, it is sufficient to prove the lemma for the case of $q=\tau$.
Let $X_l=\sum_{i=1}^{l} G_i$, where $G_i=Geom(p_i)$.
We have
\begin{eqnarray}
\sum_{i=1}^{\tau}(\sum_{j=1}^{i}\omega_j)\Pr(Y=i)&=&\sum_{j=1}^{l-1}(\sum_{j=1}^{i}\omega_j )\Pr(Y= i)
\nonumber\\      &+&\sum_{t=1}^{\tau}\Pr(X_{l}=t)(\sum_{i=l}^{\tau}(\sum_{j=1}^{i}\omega_j)\Pr(Y=i|X_{l}=t))\nonumber
\end{eqnarray}
However, since $\omega_{max}\geq \omega_l p_l\geq \omega_{l+1}p_{l+1}\geq \ldots\geq \omega_{\tau} p_{\tau}>0$, $\sum_{i=1}^{\tau}(\sum_{j=1}^{i}\omega_j)\Pr(Y=i)$ is less than the case where $\omega_l p_l=\omega_{l+1}p_{l+1}=\ldots=\omega_{\tau}p_{\tau}$. 
With a similar argument as in the proof of Theorem 1 we get
\begin{align}
&\sum_{i=1}^{\tau}(\sum_{j=1}^{i}\omega_j)\Pr(Y=i)\leq \sum_{i=1}^{l-1}(\sum_{j=1}^{i}\omega_j)\Pr(Y=i)
+ \sum_{t=1}^{\tau} (\sum_{j=1}^{l}\omega_j+\omega_l p_l(\tau-t))\Pr(X_{l}=t)\nonumber\\
&= \sum_{i=1}^{l-1}(\sum_{j=1}^{i}\omega_j)\Pr(Y=i)+( \sum_{j=1}^{l}\omega_j+\omega_lp_l\tau)(1-\Pr(X_{l}>\tau))\nonumber\\
&-\omega_lp_l[\sum_{t=1}^{\infty}t\Pr(X_l=t)-\sum_{t=\tau+1}^{\infty}t\Pr(X_l=t)]=\sum_{i=1}^{l-1}(\sum_{j=1}^{i}\omega_j)\Pr(Y=i)+( \sum_{j=1}^{l}\omega_j+\omega_lp_l\tau)\nonumber\\
&-(\sum_{j=1}^{l}\omega_j+\omega_lp_l\tau)\sum_{i=0}^{l-1}\Pr(Y=i)-\omega_lp_l\sum_{i=1}^{l}\frac{1}{p_i}+\omega_lp_l\sum_{t=\tau+1}^{\infty} t\Pr(X_{l}=t)\nonumber
\end{align}
\begin{align}
&=\sum_{i=0}^{l-1}( \sum_{j=1}^{i}\omega_j-\sum_{j=1}^{l}\omega_j-\omega_lp_l\tau)\Pr(Y=i)+(\sum_{j=1}^{l}\omega_j+\omega_lp_l(\tau-\sum_{i=1}^{l}\frac{1}{p_i}))\nonumber\\
&+\omega_lp_l\sum_{t=\tau+1}^{\infty} t\Pr(X_{l}=t)\nonumber\\
&\stackrel{(a)}{<}\sum_{i=0}^{l-1}(\sum_{j=1}^{i}\omega_j-\sum_{j=1}^{l}\omega_j-\omega_lp_l\tau)\Pr(Y=i)+(\sum_{j=1}^{l}\omega_j+\omega_{l+1})+\omega_lp_l\sum_{t=\tau+1}^{\infty} t\Pr(X_{l}=t)\nonumber
\end{align}
\begin{align}
&\stackrel{(b)}{=}\sum_{i=0}^{l-1}(-\sum_{j=i+1}^{l}\omega_j-\omega_lp_l\tau)\Pr(Y=i)+(\sum_{j=1}^{l}\omega_j+\omega_{l+1})+\sum_{i=0}^{l-1} \Pr(Y=i)(\omega_lp_l\tau+\sum_{j=i+1}^{l} \frac{\omega_lp_l}{p_j})\nonumber\\
&=(\sum_{j=1}^{l}\omega_j+\omega_{l+1})+\sum_{i=0}^{l-1}(\sum_{j=i+1}^{l} \frac{\omega_lp_l}{p_j}-\sum_{j=i+1}^{l}\omega_j)\Pr(Y=i)\stackrel{(c)}{\leq} \sum_{j=1}^{l}\omega_j+\omega_{l+1}\leq \sum_{j=1}^{l}\omega_j+\omega_{l+1}.\nonumber
\end{align}
where (a) follows from $\tau-\sum_{i=1}^{l}\frac{1}{p_i}<\frac{1}{p_{l+1}}$ and $\omega_lp_l=\omega_{l+1}p_{l+1}$; (b) follows from (\ref{equality}); and (c) follows from the fact that $\forall j\in\{i+1,\ldots ,l\}\quad \frac{\omega_lp_l}{p_j}\leq \omega_j.$
\end{proof}

Hence, by Lemma \ref{weightedlemmaup} we have
\begin{align}
w\text{-\sf T}^{\text{3}}(\eta_{\text{wg-static}}) & =\sum_{i=1}^{N}\sum_{j=1}^{q_i}(\sum_{k=1}^{j}\omega_{\mathcal I_i(k)})\Pr(Y_i=j)\stackrel{(a)}{<}\sum_{i=1}^{N}\sum_{j=1}^{l_i}\omega_{\mathcal I_i(j)}+N\omega_{max}. \nonumber\\
&\stackrel{(b)}{\leq}C_{w\text{-det}}+N\omega_{max},\nonumber
\end{align}
where (a) follows from Lemma \ref{weightedlemmaup}; and (b) follows from the fact that  $\sum_{i=1}^{N}\sum_{j=1}^{l_i}\omega_{\mathcal I_i(j)}$ is the value of the objective function in (\ref{weighteddet}) for a feasible solution.

\subsection{Proof of $C_{w\text{-det}}-2\omega_{max}\sqrt{N(C_{w\text{-det}}+\frac{N}{4})}<C_{w\text{-{\sf T}}^{\text{3}}}$}
The proof of the lower bound is  similar to the  proof of lower bound in Theorem \ref{maintheorem}. 
Consider the assignment proposed by the solution to  (\ref{weighteddet}), where the clients which have not been assigned to any {\sf AP} for transmission are assigned to {\sf AP}'s arbitrarily. 
Let $\vec\Pi_{\text{wg-static}}^{\text{det}}=[\mathcal{I}_1^{\text{det}},\mathcal{I}_2^{\text{det}},\ldots ,\mathcal{I}_N^{\text{det}}]$ denote the resulting partition, and also let $\eta_{\text{wg-static}}^{\text{det}}$ denote the corresponding static scheduling policy which orders clients based on their channel success probabilities.
Therefore, $w\text{-\sf T}^{\text{3}}(\eta_{\text{wg-static}}^{\text{det}})\leq C_{w\text{-\sf T}^{\text{3}}}.$
So, it is sufficient to prove that $C_{w\text{-det}}-2\omega_{max}\sqrt {N(C_{w\text{-det}}+\frac{N}{4})}<w\text{{-\sf T}}^{\text{3}}(\eta_{\text{wg-static}}^{\text{det}})$.

For $i\in [1:N]$ let 
$W_i\triangleq \sum_{j\in \mathcal I_i^{\text{det}}}^{}\omega_jN_j(1,\eta_{\text{wg-static}}).$
Then, by LLN we have
$w\text{-{\sf T}}^{\text{3}}(\eta_{\text{wg-static}}^{\text{det}})=\sum_{i=1}^{N}E[W_i^{\text{det}}].$
Therefore, it is sufficient to prove 
$C_{w\text{-det}}-2\omega_{max}\sqrt {N(C_{w\text{-det}}+\frac{N}{4})}<\sum_{i=1}^{N}E[W_i^{\text{det}}].$
Define $q_i=|\mathcal I_i^{\text{det}}|$,
and enumerate the clients assigned to $\text{{\sf AP}}_i$ by
$\{\mathcal I_i^{\text{det}}(1),\mathcal I_i^{\text{det}}(2),\ldots ,\mathcal I_i^{\text{det}}(q_i)\}$,
where the enumeration is according to the channel success probabilities of different clients in $\mathcal I_i^{\text{det}}$.
Further, let $G_{ij}$ be a geometric random variable with parameter $p_{ij}, i\in [1:N], j\in [1:M]$.
It is easy to see that for $k\leq q_i,i\in [1:N]$,
$W_i^{\text{det}}= \max\sum_{j=1}^{k}\omega_j   \quad s.t. \quad \sum_{j=1}^{k} G_{i\mathcal I_i^{\text{det}}(j)}\leq \tau, \quad i\in\{1,2,\ldots ,N\},k\leq q_i,$
since $\eta_{\text{g-static}}^{\text{det}}$ persistently sends a packet until it is delivered, or the interval is over.
Also define 
$l_i^{\text{det}}\triangleq \max \quad \hat l\quad s.t.\quad \sum_{j=1}^{\hat l}1/p_{i\mathcal I_i^{\text{det}}(j)}\leq \tau,\quad \hat l\leq q_i.$
Then,
\begin{align}
\sum_{i=1}^{N}E[W_i^{\text{det}}]  & \stackrel{(a)}{>}   \sum_{i=1}^{N}\sum_{j=1}^{l_i^{\text{det}}} \omega_{\mathcal I_i^{\text{det}}(j)}-2\omega_{max}\sum_{i=1}^{N}\sqrt{l_i^{\text{det}}+\frac{1}{4}}\nonumber\\
&\stackrel {(b)}{\geq}   \sum_{i=1}^{N}\sum_{j=1}^{l_i^{\text{det}}} \omega_{\mathcal I_i^{\text{det}}(j)}-2\omega_{max}\sqrt {N(\sum_{i=1}^{N}l_i^{\text{det}}+\frac{N}{4})}\nonumber\\
& \geq \sum_{i=1}^{N}\sum_{j=1}^{l_i^{\text{det}}} \omega_{\mathcal I_i^{\text{det}}(j)}  -2 \omega_{max}\sqrt {N(\sum_{i=1}^{N}\sum_{j=1}^{l_i^{\text{det}}} \omega_{\mathcal I_i^{\text{det}}(j)}+\frac{N}{4})}\nonumber\\
&\stackrel{(c)}{=} C_{w\text{-det}}-2\omega_{max}\sqrt {N(C_{w\text{-det}}+\frac{N}{4})},\nonumber
\end{align}
where (a) follows from Lemma \ref{weightedlemmadown}; 
(b) follows from Cauchy-Schwarz inequality;
and (c) follows from $\sum_{i=1}^{N}\sum_{j=1}^{l_i^{\text{det}}} \omega_{\mathcal I_i^{\text{det}}(j)}=C_{w\text{-det}}$.
Hence, the left inequality of Theorem \ref{weightedtheorem} is proved and the proof of Theorem \ref{weightedtheorem} is complete. 

\begin{lemma}\label{weightedlemmadown}
Let $1\leq\omega_1,\omega_2,\ldots ,\omega_q\leq \omega_{max}$ for some $\omega_{max}\in\mathbb R$.
Also, let $\tau \in \mathbb N$ and $G_1, G_2, \ldots, G_q$ be independent geometric random variables with parameters $p_1, p_2, \ldots, p_q$ respectively, such that $1\geq p_1 \geq p_2 \geq \ldots \geq p_q \geq 0$. Also define
$l\triangleq \max \hat l\quad s.t.\quad \sum_{i=1}^{\hat l}1/p_i\leq \tau,$
 and
$Y\triangleq \max i\quad s.t. \quad \sum_{j=1}^{i} G_j\leq \tau, \quad i\in\{1,2,\ldots ,q\}.$
 Then, we have
$\sum_{j=1}^{l}\omega_j-2\omega_{max}\sqrt{l+\frac{1}{4}}<\sum_{i=1}^{q}(\sum_{j=1}^{i}\omega_j)\Pr(Y=i).$
\end{lemma}

\begin{proof}
With the same argument as in Theorem 1, it is sufficient to assume $q=l$.
The proof is very similar to the proof of lower bound in Theorem 1:
\begin{align}
&\sum_{i=1}^{l}\omega_i-\sum_{i=1}^{l}(\sum_{j=1}^{i}\omega_j)\Pr(Y=i)=
\sum_{i=1}^{l}\omega_i-\sum_{i=1}^{l}\omega_i(\sum_{j=i}^{l}\Pr(Y=j))\nonumber\\
&=\sum_{i=1}^{l}\omega_i\Pr(G_1+G_2+\ldots G_i>\tau)\nonumber
\end{align}
\begin{align}
&\leq\sum_{i=1}^{l-1}\omega_i\Pr(|\sum_{j=1}^{i} (G_j-\frac{1}{p_j})|>\sum_{j=i+1}^{l}\frac{1}{p_j})+\omega_l\stackrel{(a)}{\leq} \sum_{i=1}^{l-1}\omega_i\min(1,\frac{var(\sum_{j=1}^{i} G_j)}{(\sum_{j=i+1}^{l}\frac{1}{p_j})^2})+\omega_l\nonumber\\
&\leq\sum_{i=1}^{l-1}\omega_i\min(1,\frac{\sum_{j=1}^{i}\frac{1}{p_j^2}}{(\sum_{j=i+1}^{l}\frac{1}{p_j})^2})+\omega_l\nonumber\leq\omega_l+\sum_{i=1}^{l-1}\omega_i\min(1,\frac{i}{(l-i)^2})\nonumber\\
&\leq\omega_{max}(1+\sum_{i=1}^{l-1}\min(1,\frac{i}{(l-i)^2}))\stackrel{(b)}{\leq} 2\omega_{max}\sqrt {l+\frac{1}{4}},\nonumber
\end{align}
where (a) follows from Chebyshev's inequality; and (b) follows from Lemma \ref{ell}.
\end{proof}

\end{document}